\documentclass[11pt,a4paper]{amsart}
\usepackage[dvipdfm]{graphicx}
\usepackage[hidelinks]{hyperref}
\usepackage{tabstackengine}
\usepackage{verbatim}
\usepackage{amssymb}
\usepackage{amsbsy}
\usepackage{amscd}
\usepackage{amsmath}
\usepackage{amsthm}
\usepackage{amsxtra}
\usepackage{geometry}
\usepackage{latexsym}
\usepackage[mathscr]{eucal}
\usepackage{tikz}
\usepackage{filecontents}
\usetikzlibrary{arrows}
\theoremstyle{plain}
\newtheorem*{theorem*}{Theorem}
\newtheorem{thm}{Theorem}[section]
\newtheorem{lemm}[thm]{Lemma}
\newtheorem{dfn}[thm]{Definition}

\newtheorem{rmk}[thm]{Remark}

\newtheorem*{FTA}{Factorization theorem A}
\newtheorem*{FTB}{Factorization theorem B}
\newtheorem*{IVL}{Involution theorem}
\newtheorem{cor}[thm]{Corollary}
\newtheorem*{cor*}{Corollary}
\newtheorem{notation}[thm]{Notation}
\newtheorem{exa}[thm]{Example}

\newtheorem{conv}[thm]{Convention}
\newcommand{\C}{\mathbb{C}}
\newcommand{\Z}{\mathbb{Z}}
\newcommand{\e}{\ell}
\newcommand{\Span}{\mathrm{Span}}
\newcommand{\End}{\mathrm{End}}
\newcommand{\Hom}{\mathrm{Hom}}
\newcommand{\grF}{\mathrm{gr}^{F}}

\newcommand{\Jm}{J_{m}}
\newcommand{\Jinf}{J_{\infty}}
\newcommand{\Spec}{\mathrm{Spec}}
\newcommand{\FG}{\mathfrak{g}}
\newcommand{\fraksl}{\mathfrak{sl}}
\newcommand{\nondp}{(~\mid~)}
\newcommand{\Ei}{\mathrm{Ei}}

\newcommand{\sech}{\mathrm{sech}}
\newcommand{\W}{\mathscr{W}}
\newcommand{\V}{\mathcal{V}}
\newcommand{\Sym}{\mathrm{Sym}}
\newcommand{\kmg}{\widehat{\mathfrak{g}}}
\title{Generalized AKS scheme of integrability via vertex algebra}
\author{Wenda Fang}
\address{Research Institute for Mathematical Sciences, Kyoto University,
 Kyoto 606-8502 JAPAN}

\email{wenda@kurims.kyoto-u.ac.jp}

\thanks{This work was supported by JST SPRING, Grant Number JPMJSP2110.}


\begin{document}
\maketitle
\begin{abstract}
In this paper, we define and study the classical $R$-matrix for vertex Lie algebra, based on which we propose to construct a new vertex Lie algebra. We give a systematic way to construct the $R$-matrix for affine Kac-Moody vertex Lie algebra and study the universal vertex algebra associated with the new vertex Lie algebra that we obtained by $R$-matrix. As an application, using the classical $R$-matrix we defined, we give a new scheme to construct infinite-dimensional (Liouville) integrable systems via the Feigin-Frenkel center. The scheme-theoretical explanation of our equations and the classical $\W$-algebra case of our theory will come later.
\end{abstract}
\section{Introduction}
This is the first one of a series of papers on studying the connection between vertex algebras and infinite-dimensional Hamiltonian integrable systems via the $R$-matrix method.
\subsection{Classical R-matrix for the vertex Lie algebra} The classical $R$-matrix of Lie algebras appeared in Skylanin's works as a ``by-product" of the quantum inverse-scattering method \cite{KBI}. It also appeared in the earlier paper by Michaelis \cite{M1}. Later on, in the paper \cite{Se2} Semenov-Tyan-Shanskii established the connection between the classical $R$-matrices and integrable systems \cite{Se2}. Let $\FG$ be a Lie algebra with non-degenerate, $ad$-invariant, bilinear form $\nondp$. We say $R\in\End\FG$ is a classical $R$-matrix of $\FG$ if the new bracket defined by
\[[X,Y]_{R}:=[R(X),Y]+[X,R(Y)],\ \forall X,Y\in\FG,\]
together with the vector space $\FG$ forms a Lie algebra.
In \cite{Se2}, Semenov-Tyan-Shanskii proved a sufficient condition for an $R\in\End\FG$ to be an $R$-matrix of a Lie algebra $\FG$ called the modified Yang-Baxter equation. The classical Yang-Baxter equation in the \cite{Se2} has been explained as a ``unitary" condition of the $R$-matrix. This motivated numerous works in different areas, e.g. quantum groups \cite{ChaP}, Poisson-Lie groups\cite{YKSch}. In addition, he built the connection between the classical $R$-matrix and the Adler-Kostant-Symes scheme, which is a powerful tool to construct integrable systems such as the Toda lattice, etc (it is also possible to construct the KdV equation via the classical AKS scheme \cite{Ad}), see Section \ref{subsec:PVAAHS}. Also, in the same paper, Semenov-Tyan-Shanskii showed the connection between the Riemann-Hilbert problem and the $R$-matrix method. In the paper \cite{RS}, Reyman and Semenov-Tyan-Shanskii gave a geometric interpretation of the AKS scheme and the classical $R$-matrix. The classical $R$-matrix theory also appeared in the physical understanding of integrability, see \cite{BS,CWY1,CWY2}.\par
The vertex algebras were introduced by Borcherds \cite{Bo1} and can be viewed as an axiomatic approach to 2d chiral conformal field theory. In the 2d case, thanks to the conformal symmetry, the state space of our quantum field theory becomes particularly simple. In many cases, it becomes the highest weight module for certain infinite-dimensional Lie algebra, e.g., Kac-Moody algebra, etc., see \cite{BPZ,Sch}. Vertex algebra naturally produces vertex Lie algebra and vice versa (for precise definition see Section \ref{secpre1}). From the algebraic viewpoint, the vertex (operator) algebras are analogous to both Lie algebras and commutative associative algebras. Hence, it is natural to generalize the classical $R$-matrix to the vertex algebra case. X. Xu defined the classical $R$-matrix for vertex operator algebra \cite{Xu1}, which is a kind of special vertex algebra with a conformal vector. However, it seems difficult to construct examples by his approach. In contrast, we may construct many examples through our construction.\par
As an analog of Semenov-Tyan-Shanskii's work in the paper \cite{Se2}, we obtain some similar results called factorization theorem A and factorization theorem B, which we stated in Section \ref{sec:PFT}. Besides this, in Section \ref{sec:MYBE}, we show that for a vertex Lie algebra $L$, an operator $R_L\in\End L$, if $R_L$ satisfies the modified Yang-Baxter equation \eqref{MCYB} of vertex Lie algebra then $R_L$ is a classical $R$-matrix of $L$. In particular, factorization theorem B gives an effective way to construct the $R$-matrix for the affine vertex Lie algebra. Here is an example, consider a Lie algebra $\FG$ with a non-degenerate bilinear $ad$-invariant form $\nondp$ and admit the Lie algebra decomposition $\FG=\mathfrak{a}\oplus\mathfrak{b}$ where $\mathfrak{b}$ is an isotropic subalgebra of $\FG$. This gives a vertex Lie algebras factorization automatically, i.e., $v_{k}(\FG)=v_k(\mathfrak{a})\oplus v_k'(\mathfrak{b})$ (see Example \ref{exa:AVLA} and Theorem \ref{thm:RM} for definitions). Let $P_+,\ P_-$ be two projections from $v_k(\FG)$ to $v_k(\mathfrak{a}),\ v_k'(\mathfrak{b})$, respectively. Then the operator $R_{L}:=\frac{1}{2}(P_+-P_-)$ is an $R$-matrix of the vertex Lie algebra $v_k(\FG)$.\par
One of the applications of the $R$-matrix of vertex Lie algebra is that it can be used to construct new vertex algebras. We calculate the examples in Section \ref{subsec:ECR}. A sensible idea is to construct the factorizable $R$-matrix (see Theorem \ref{thm:RM} and Corollary \ref{cor:FRMT}) via the Manin triple of Lie algebra, which possibly gives a connection between our theory with quantum groups.
\subsection{Poisson vertex algebra and Hamiltonian systems}\label{subsec:PVAAHS}
From the viewpoint of mathematical physics, the Poisson vertex algebra (PVA) may obtain by taking the (quasi-)classical limit of a vertex algebra \cite{DSK06,K1,Li1}. More mathematically, it is a Poisson analogue of a Lie conformal algebra. Since any vertex algebra is filtrated and its structure is quite complicated, in many cases, people study the geometry of the associated vertex Poisson scheme, i.e., the spectrum of associated Poisson vertex algebra, see \cite{Ara2,Ara3,AraM,Fre}.\par
One of the origins of the PVA is as follows. There is a long-stand problem 
stated by A. Belavin in the paper \cite{B1}, which asks about the connection between $\W$-algebras and the integrable systems, in particular the KdV equation. The same ideas also appeared in many other previous works, for example \cite{DG}. This problem has been completely solved by V. Kac and his collaborators in a series of papers, see \cite{DSKV13}. In particular, in the paper \cite{BDSK09}, Kac and his collaborators showed that the Poisson vertex algebra plays an important role in the infinite-dimensional Hamiltonian integrable systems. Many connections between PVA, Dubrovin-Novikov theory \cite{DN1,DN2}, and Dubrovin-Zhang theory \cite{DZ} have been indicated in the papers 
 \cite{Ca1,Ca15}.\par
Another goal of this paper is to generalize the AKS scheme by using the classical $R$-matrix we defined. Following \cite{Se2}, we explain the main observation of the classical AKS scheme and how it connects to the classical $R$-matrix, see also \cite{Kir1,Ber1,AM}. Consider a Lie algebra $\FG$ with a non-degenerate, $ad$-invariant bilinear form $(\cdot,\cdot)$ and an $R$-matrix on $\mathfrak{g}$, then we have two Poisson structures on $\mathfrak{g}^{*}$ given by the Lie brackets $[\cdot,\cdot]$ and $[\cdot,\cdot]_{R}$ known as the Poisson-Lie-Kirillov-Kostant Poisson structure. Denote the corresponding Poisson brackets by $\{\cdot,\cdot\}$ and $\{\cdot,\cdot\}_{R}$. The key observation of the AKS scheme is that the elements in the Poisson center with respect to $\{\cdot,\cdot\}$ Poisson commute with respect to $\{\cdot,\cdot\}_{R}$. Therefore, we obtain a family of the first integrals. If we restrict our attention to the algebraic function on $\mathfrak{g}^{*}$, then $\mathbb{C}[\mathfrak{g}^{*}]\cong S(\mathfrak{g})$. We see that as a physical system on $\mathfrak{g}^{*}$, the functions of physical observables lie in $S(\mathfrak{g})$. This observation gives an effective way to construct integrable systems. In addition, it has been used to explain various aspects of soliton equations, see \cite{TU}.\par
Poisson vertex algebra gives a framework of infinite-dimensional integrable systems (field-theoretic systems), see Section \ref{secpre2}. Let $L$ be a vertex Lie algebra with a $\lambda$-bracket $\{\cdot_{\lambda}\cdot\}$ and a classical $R$-matrix $R_L$. Then $L$ has two $\lambda$-brackets, $\{\cdot_{\lambda}\cdot\}$ and $\{\cdot_{\lambda}\cdot\}_{R_{L}}$, see Section \ref{subsec:CRLCRB}.
Consider the Poisson vertex algebra $\V:=\Sym(L)$ and extend the $\lambda$-brackets $\{\cdot_{\lambda}\cdot\}$ and $\{\cdot_{\lambda}\cdot\}_{R_{L}}$ via the Master formula \eqref{MF}. This gives two $\lambda$-brackets on $\V$.
\begin{IVL}\label{thm7.1}
Let $L$ be a vertex Lie algebra with a $\lambda$-bracket $\{\cdot_{\lambda}\cdot\}$, $R_{L}$ be an $R$-matrix of $L$, $\mathcal{V}:=\Sym(L)$, $\mathfrak{z}(\V)$ be the center of $\V$ with respect to $\{\cdot_{\lambda}\cdot\}$. Then $\forall f,g\in\mathfrak{z}(\V),\ \{f_{\lambda}g\}_{R_{L}}=0$.
\end{IVL}
This gives a way to construct the Liouville integrable systems. Let $\V,\{\cdot_{\lambda}\cdot\},\ \{\cdot_{\lambda}\cdot\}_{R_{L}}$ as in the involution theorem. Choose a family of linearly independent polynomials $h_n\in\mathfrak{z}(\V),\ n=1,2,\cdots$ and take functionals $\int h_n\in\V /\partial\V$ as Hamiltonians. We get non-trivial integrable systems through the twisted $\lambda$-bracket $\{\cdot_{\lambda}\cdot\}_{R_{L}}$, see Definition \ref{dfn:HEQ}.\par
Due to the existence of the Feigin-Frenkel center \cite{FF}, in the affine vertex algebra case we have the following corollary.
\begin{cor*}
Consider affine vertex algebra $V^{k}(\FG)$ at the critical level, i.e., $k=-h^{\vee}$ where $h^{\vee}$ is the dual Coxter number of $\FG$. Let $\mathfrak{z}(\widehat{\FG})$ be the Feigin-Frenkel center contained in the $V^{-h^{\vee}}(\FG),\ \mathcal{V}:=\grF V^{k}(\FG)\cong\Sym'(v_0(\FG))$ as Poisson vertex algebras, see Remark \ref{rmk:level0}. Suppose $R_{L}$ be an $R$-matrix of $v_{0}(\FG)$, then $\forall f,g\in\grF\mathfrak{z}(\widehat{\FG}),$ $\{f_{\lambda}g\}_{R_{L}}=0$.
\end{cor*}
We give a specific example here. Take the affine vertex algebra $V^{-2}(\mathfrak{sl}_2)$, the $R$-matrix of $v_0(\mathfrak{sl}_2)$ associated with the Borel decomposition (see Definition \ref{dfn:RM}), and apply our construction, one obtains 
that the following system
\begin{equation*}
\begin{aligned}
&\frac{d f}{dt_2}=0,\\
&\frac{d h}{dt_2}=-8 fe({h^{\prime}}^2+4 f^{\prime} e^{\prime}+2 e f^{\prime\prime}+hh^{\prime\prime}+2 f e^{\prime\prime}),\\
&\frac{d e}{d t_2}=4 he({h^{\prime}}^2+4 f^{\prime}e^{\prime}+2 e f^{\prime \prime}+hh^{\prime\prime}+2f e^{\prime\prime}),
\end{aligned}
\end{equation*}
where the primes are the derivative with respect to displacement variable $x$ is Liouville integrable. In Example \ref{exa:BSL2O}, we show that one solution of this system is related to the KdV one soliton. Also, in the same example, we show that the solutions of this system are in one-to-one correspondence with the points on $\mathbb{CP}^{1}$.\par
Our approach extends the $R$-matrix method in the Lie algebras to the vertex algebras and in a special case, we recover the classical case. Same as the classical AKS scheme our scheme has geometric meaning that we will explain in our further work. In addition, this approach is possibly generalized to the classical $\W$-algebra and related to the Khesin-Zakharevich Poisson-Lie group \cite{LM}. The author also hopes this paper and further works will explain the integrable systems, in particular soliton equations, in a systematic way.\par
This paper is organized as follows. In Section \ref{secpre1} we shall refer to notations and results given in \cite{FBZ,K1}. In Section \ref{secpre2}, we summarize how Poisson vertex algebras are related to the integrable systems. In Section \ref{sec:MYBE}, we prove two sufficient conditions for an operator $R_{L}\in \mathrm{End} L$ to be an $R$-matrix. In Section \ref{sec:PFT}, we prove our factorization theorem for vertex Lie algebra, also in Section \ref{subsec:ECR}, we show that the universal affine vertex algebra of $v_{k}(\mathfrak{g})_{R_{L}}$ associated to the Iwasawa decomposition is simple. In Section \ref{sec:AKS}, we generalize the Adler-Kostant-Symes scheme into the vertex algebra case. In Section \ref{sec:app}, we compute a few examples given by our construction.\\
{\it Acknowledgements} \\ 
This paper is the master thesis of the author, and he
wishes to express his gratitude to his supervisor Tomoyuki Arakawa for lots of advice to improve this paper. He especially thanks Prof. Takahiro Shiota for useful discussions and many suggestions on this paper. Also, he thanks Vladimir E. Zakharov for introducing the integrable systems to him. He is deeply grateful to Sun Furihata, and Shoma Sugimoto for their many pieces of advice and encouragement. He thanks Ryo Sato, and Shigenori Nakatsuka for useful comments and discussions. He appreciates that Dr. Xuanzhong Dai and Professor Marek Rychlik read his manuscript carefully and give many advice. He thanks Dr. Yuto Moriwaki for the opportunity to visit the RIKEN and his valuable comments on this paper.
\begin{conv}
In this paper, our ground field will be $\C$. In addition, the word ``classical" in the classical AKS scheme does not state the counterpart of the quantum case, but the Lie algebraic version of the AKS scheme.
\end{conv}
\section{Preliminaries I}\label{secpre1}
\subsection{Vertex algebras and vertex Lie algebras} In this part we recall the definitions of vertex algebras, all concepts can be found in \cite{FBZ} and \cite{LL}. 
\begin{dfn}\label{def:VA}
	A vertex algebra is a vector space $V$ equipped with a vector $|0\rangle$ (called the vacuum vector), a linear map 
\begin{align*}
Y(?,z):V\longrightarrow\End\ V[[z^{\pm}]],\ a\longmapsto Y(a,z)=\sum_{n\in\mathbb{Z}}a_{(n)}z^{-n-1},
\end{align*} and a linear operator $T\in\End\ V$. These data are subject to the following axioms:
\begin{itemize}
\item (vacuum axioms) $Y(|0\rangle,z)=Id_{V}$. For $a\in V,$
$$Y(a,z)|0\rangle\in V[[z]],$$
and $Y(a,z)|0\rangle\mid_{z=0}=a$.\\
\item (translation axiom) For any $A\in V$,
$$[T,Y(a,z)]=\partial_{z}Y(a,z),$$
and $T|0\rangle=0$.\\
\item (Locality axiom) $\forall a,b\in V,$ $Y(a,z),\ Y(b,z)$ are locally mutual, i.e., $\exists N_{ab}\in\mathbb{Z}_{>0}$ s.t. 
$$(z-w)^{N_{ab}}[Y(a,z),Y(b,w)]=0.$$
\end{itemize}
\end{dfn}
The vacuum axiom is equivalent to $a_{(n)}|0\rangle=0,\ \forall n\geq0$ and $a_{(-1)}|0\rangle=a$. Also, for simplicity, we may denote $Y(a,z)$ by $a(z)$.
\begin{exa}\label{exa:VA}
Let $\mathfrak{g}$ be a finite-dimensional Lie algebra with normalized invariant bilinear form $(\cdot\mid\cdot)$. The Kac-Moody affinization is written as $\widehat{\mathfrak{g}}=\mathfrak{g}[t,t^{-1}]\oplus\mathbb{C}K$. The commutator on $\widehat{\mathfrak{g}}$ is given by formula:
$$[xt^{n},yt^{m}]=[x,y]t^{n+m}+n(x\mid y)\delta_{n+m,0}K,\ [\widehat{\mathfrak{g}},K]=0.$$
Let $k$ be a complex number, the universal affine vertex algebra associated with $\mathfrak{g}$ at a level $k$ is defined as
$$V^{k}(\mathfrak{g})=U(\widehat{\mathfrak{g}})\otimes_{U(\mathfrak{g}[t]\oplus\mathbb{C}K)}\mathbb{C}_{k},$$
where $\mathbb{C}_{k}$ is the one-dimensional representation of $\mathfrak{g}[t]\oplus\mathbb{C}K$ on which $\mathfrak{g}[t]$ acts trivially and $K$ acts as a multiplication by $k$. Then $V^{k}(\mathfrak{g})$ has a natural vertex algebra structure \cite{FZh}. We shall call it the universal affine vertex algebra associated with $\mathfrak{g}$ at level $k$. It is customary to denote the vacuum vector in this case by $v_{k}$.
\end{exa}
Let $V$ be a vertex algebra, and $I\subset V$. We say I is a vertex  algebra ideal if it is a $T$-invariant subspace satisfying $Y(a,z)b\in I((z)),\ \forall a\in I,\ b\in V$. In addition, we say the vertex algebra is simple if $V$ has no proper ideal.\par
A vertex algebra $V$ has skew-symmetry property, i.e., $Y(A, z) B=e^{z T} Y(B,-z) A$ in $V((z))$. 
Hence, all vertex algebra ideals are two-sided ideals. Together with the vacuum axiom, we conclude that if $|0\rangle\in I,$ then $I$ is not a proper ideal.
\begin{dfn}\label{dfn:CVA}

The center $\mathfrak{z}(V)$ of a vertex algebra V is a subspace defined by
$$\mathfrak{z}(V)=\{b\in V\mid a_{(n)}b=0\ \forall a\in V,\ n\gg0\}.$$ Equivalently, $b\in\mathfrak{z}(V)$ if and only if $[Y(a,z),Y(b,w)]=0\ ,\forall a\in V$.
\end{dfn}
It is well-known that for the affine vertex algebra $V^{k}(\mathfrak{g})$,  $\mathfrak{z}(V^{k}(\mathfrak{g}))=\mathbb{C}|0\rangle$, unless $k=-h^{\vee}$ where $h^{\vee}$ is the dual Coxeter number of $\mathfrak{g}$. Hence, we introduce the following definition.
\begin{dfn}[\cite{Ha1,FF}]\label{dfn:FFC}
The center of the affine vertex algebra $V^{-h^{\vee}}(\mathfrak{g})$ is called the Feigin-Frenkel center and denoted by $\mathfrak{z}(\widehat{\mathfrak{g}})$. In addition, the elements that lie in the Feigin-Frenkel center are called the Segal-Sugawara operators.
\end{dfn}
\begin{exa}\label{exa:SO}
For the affine vertex algebra $V^{-2}(\mathfrak{sl}_{2})$, let $S=e_{-1}f_{-1}+f_{-1}e_{-1}+\frac{1}{2}h_{-1}h_{-1}$ then
$$\mathfrak{z}(\widehat{\mathfrak{sl}}_{2})=\mathbb{C}[T^{r}S\mid r\geq0].$$
\end{exa}
\begin{dfn}\label{def2.1}
 A vertex Lie algebra is a vector space $L$, together with a linear operator 
 \begin{equation*}
 	Y^{-}(?,z):L\longrightarrow\mathrm{End}L\otimes z^{-1}\mathbb{C}[[z^{-1}]],\ a\longmapsto\sum_{n\geq0}a_{(n)}z^{-n-1},
 \end{equation*}
where the operators $a_{(n)}$'s satisfy $a_{(n)}v=0\ \forall n\ll0$, and a linear operator $T$ on $L$.
These data satisfy the following axioms:
\begin{itemize}
\item (translation) $Y^{-}(Ta,z)=\partial_{z}Y^{-}(a,z)$,\\
\item (skew-symmetry) $Y^{-}(a,z)b=(e^{zT}Y^{-}(b,-z)a)_{-},$\\
\item (commutator) 
\begin{equation*}
[a_{(m)},Y^{-}(b,w)]=\sum_{n\geq0}\binom{m}{n}(w^{m-n}Y^{-}(a_{(n)}b,w))_{-},
\end{equation*}
where for a power series $a(z)=\sum_{m\in\mathbb{Z}}a_{m}z^{m},$ we define $a(z)_{-}=\sum_{m<0}a_{m}z^{m}.$
\end{itemize}
\end{dfn}
\begin{rmk}
It follows from Haisheng Li's paper \cite{Li1} that the commutator is equivalent to the following half Jacobi identity: $\forall u,v\in L$,
\begin{equation}\label{Eq:HJI}
\begin{aligned}
(&x^{-1}\delta(\frac{y-z}{x})Y^{-}(u,y)Y^{-}(v,z)-x^{-1}\delta(\frac{z-y}{-x})Y^{-}(v,z)Y^{-}(u,x))_{-}\\
&=(z^{-1}\delta(\frac{y-x}{z})Y^{-}(Y^{-}(u,x)v,z))_{-},
\end{aligned}
\end{equation}
where $\delta(z):=\sum_{n\in\mathbb{Z}}z^{n}$ and whenever $(z-w)^{m}$ form appears, it is understood to be the binomial expansion in non-negative power of $w$.
\end{rmk}
\begin{rmk}\label{rmk:lambda}
Let $(V,Y^{-}(?,z),T)$ be a vertex Lie algebra. For all $a,b\in V$, define $$[a_{\lambda}b]:=\sum_{n\geq0}\frac{\lambda^{n}}{n}a_{(n)}b.$$
This is called $\lambda$-bracket. In addition, one re-writes Definition \ref{def2.1} in terms of $\lambda$-bracket. More precisely, the translation, skew-symmetry, and the Jacobi identity are equivalent to 
\begin{align*}
&[T a_\lambda b]=-\lambda[a_\lambda b], \quad[a_\lambda T b]=(\lambda+T)[a_\lambda b], \quad \mathrm{(sesquilinearity)}  \\
&[a_\lambda b]=-[b_{-\lambda-T} a], \mathrm{(skew symmetry ) } \\
&[a_\lambda[b_\mu c]]-[b_\mu[a_\lambda c]]=[[a_\lambda b]_\mu c],\mathrm{(Jacobi\ identity)}
\end{align*}
respectively. Then $(V,[\cdot_{\lambda}\cdot],T)$ forms a Lie conformal algebra (for definition see \cite{K1}). 
\end{rmk}
One can construct vertex Lie algebra from vertex algebra. For a vertex algebra $(V,|0\rangle,T,Y(?,z)),$ $\forall A\in V$ we denote
	$$Y^{-}(A,z):=\sum_{n\geq0}A_{(n)}z^{-n-1}.$$
Then $(V,Y^{-}(?,z),T)$ forms a vertex Lie algebra, see {\cite[Lemma 16.1.3]{FBZ}}.
\begin{exa}\label{exa:AVLA}
Consider the affine vertex algebra defined in Example \ref{exa:VA}. We denote the vector space spanned by the vacuum vector $v_{k}\in V^{k}(\mathfrak{g})$ and $xt^{n}.v_{k},\ n\leq-1$ by $v_{k}(\mathfrak{g})$. Then $(v_{k}(\mathfrak{g}),Y^{-}(?,z),T)$ has vertex Lie algebra structure. We will call the vertex Lie algebra $v_k(\FG)$ the affine vertex Lie algebra of Lie algebra $\FG$ at level $k$.
\end{exa}
Since we will use the notation of vertex Lie subalgebra, we give an example here.
\begin{exa}\label{exa:VLSA}
Consider vertex Lie algebra $(v_{k}(\mathfrak{sl}_2),Y^{-}(?,z),T)$ and let $f,h,e$ be the standard generators of $\mathfrak{sl}_2$. Then the subspace of $v_k(\mathfrak{sl_2})$ spanned by $v_k,h_{-1}.v_k,\ e_{-1}.v_k$ forms a vertex Lie subalgebra of $v_{k}(\mathfrak{sl}_2)$.
\end{exa}
\begin{dfn}\label{def2.2}
Let $(L_{1},Y_{1}^{-},T_{1})$ and $(L_{2},Y^{-}_{2},T_{2})$ be vertex Lie algebras. A linear map $\varphi:L_{1}\longrightarrow L_{2}$ is called a weak vertex Lie algebra homomorphism, if
\begin{equation*}
	[T,\varphi]=0,\ \varphi[Y^{-}_{1}(u,z)v]=Y_{2}^{-}(\varphi(u),z)\varphi(v),\ \forall u,v\in L_{1}.
\end{equation*}
\end{dfn}
Note that by [\cite{LL}, Proposition 3.12.1] if $(L_{1},Y_{1}^{-},T_{1})$ and $(L_{2},Y_{2}^{-},T_{2})$ is vertex Lie algebras, then $(L_{1}\oplus L_{2},Y^{-}_{1}\oplus Y^{-}_{2},T_{1}\oplus T_{2})$	is also a vertex Lie algebra.\par
Follows \cite{FBZ}, let $\mathrm{L}$ be a vertex Lie algebra and denote $L(D^{\times}):=L\otimes\mathbb{C}((t)).$ Consider the operator $\partial:=T\otimes1+Id\otimes\partial_{t}$ acting on $L(D^{\times}),$ then 
$$\mathrm{Lie}(L):=L(D^{\times})/\mathrm{Im}(\partial),$$ forms a Lie algebra called local Lie algebra of $L$. For all $a\in L$ we denote the image of $a\otimes t^{n}$ in $\mathrm{Lie}(L)$ by $a_{[n]}$. Then the commutator on $\mathrm{Lie}(L)$ reads as:
\begin{equation}\label{eq:LLA}
[a_{[m]},b_{[k]}]=\sum_{n\geq0}\binom{m}{n}(a_{(n)}b)_{[m+k-n]}.
\end{equation}
Let $\mathrm{Lie}(L)_{+}$ be the completion of the span of all $a_{[n]},\ n\geq0$, then $\mathrm{Lie}(L)_{+}$ is a Lie subalgebra of $\mathrm{Lie}(L)$ with respect to the Lie bracket defined by \eqref{eq:LLA}.
\begin{exa}\label{exa:LLA}
Take affine vertex Lie algebra $v_{k}(\mathfrak{g})$ as in Example \ref{exa:AVLA}, then $\mathrm{Lie}(v_{k}(\mathfrak{g}))\cong\widehat{\mathfrak{g}}$ as Lie algebras.
\end{exa}
By {\cite[Lemma 16.1.7]{FBZ}}, for a vertex Lie algebra we can construct a vertex algebra. Consider a vertex Lie algebra $\mathrm{L}$. Let $U(L)$ be the universal enveloping algebra of $\mathrm{Lie}(L)$ and $U(L)_{+}$ be the universal enveloping algebra of $\mathrm{Lie}(L)_{+}$, then 
$$\mathrm{Vac}(L):=U(L)\otimes_{U(L)_{+}}\mathbb{C},$$
where $\mathbb{C}$ is the trivial one-dimensional representation of $\mathrm{Lie}(L)_{+}$, and the vacuum vector is $1\otimes1$. For all $a\in L$, the field is given by formula $$Y(a_{[-1]}|0\rangle,z)=\sum_{n\in\mathbb{Z}}a_{[n]}z^{-n-1},$$
where $|0\rangle:=1\otimes1$. 
For example, let vertex Lie algebra $L$ be $v_{k}(\mathfrak{g}),$ then it follows from Example \ref{exa:LLA} that $\mathrm{Vac}(L)\cong V^{k}(\mathfrak{g})$ as vertex algebras.
\subsection{Poisson vertex algebras}
\begin{dfn}
A Poisson vertex algebra is a quintuple $(\mathscr{V},\partial,1,\cdot,\{\cdot_{\lambda}\cdot\})$ such that $(\mathscr{V},1,\partial)$ is a commutative associative unital differential algebra and $(\mathscr{V},\partial,\{\cdot,\cdot\})$ is a Lie conformal algebra where $\lambda$-bracket is given by $\{\cdot_{\lambda}\cdot\}$, and such that the $\lambda$-bracket $\{\cdot_{\lambda}\cdot\}$ and the associative product $``\cdot"$ satisfy
\begin{enumerate}\nonumber
\item(left Leibniz rule) $\{a_{\lambda}bc\}=\{a_{\lambda}b\}c+b\{a_{\lambda}c\},$
\item(right Leibniz rule) $\{ab_\lambda c\}=\{a_{\lambda+\partial}c\}_\rightarrow b+\{b_{\lambda+\partial}c\}_\rightarrow a$,
\end{enumerate}
where the right arrow and $\lambda$-bracket are understood as the following. If $\{a_{\lambda}b\}=\sum_{n\in\mathbb{Z}_{+}}\frac{\lambda^{n}}{n!}a_{(n)}b,$ when a right arrow appears, that means that $\lambda+\partial$ has to be moved to the right i.e., $\{a_{\lambda+\partial}b\}_{\rightarrow}=\sum_{n\in\mathbb{Z}_{+}}\frac{a_{(n)}b}{n!}(\lambda+\partial)^{n}$. If no right arrow appears, then $\{a_{-\lambda-\partial}b\}=\sum_{n\in\mathbb{Z}_{+}}\frac{(-\lambda-\partial)^{n}}{n!}a_{(n)}b$.
\end{dfn}
\begin{exa}\label{exa:pva}
Consider the vertex Lie algebra given in Example \ref{exa:AVLA} $(L,\{\cdot_{\lambda}\cdot\},\partial)$. Let $\mathscr{V}:=\Sym(L)$ i.e., the symmetric algebra of vector space $L$. If we extend the $\lambda$-bracket by the left and right Leibniz rule, then the symmetric algebra $\mathscr{V}$ has a Poisson vertex algebra structure.
\end{exa}
According to Proposition 3.8 in the paper \cite{Li1}, we know for an affine vertex Lie algebra $v_k(\FG)$, the quotient space $\mathrm{Sym}'(v_k(\FG)):=\mathrm{Sym}(v_k(\FG))/(1-v_k)$ is also a Poisson vertex algebra. In this paper, we will call it the modified Poisson vertex algebra obtained from $v_k(\FG)$.
\begin{rmk}
The right Leibniz rule above can be obtained from the left Leibniz rule and the skewsymmetry, cf. \cite{K1} Exercise 13.
\end{rmk}
\begin{rmk}
One defines the Poisson vertex algebra through vertex Lie algebra by the following way. A Poisson vertex algebra is a differential algebra $(V,\partial)$ equipped with a vertex Lie algebra structure $(Y^{-}(?,z),\partial)$ such that $\forall a,b,c\in V,$
$$Y^{-}(a, x)(b c)=\left(Y^{-}(a, x) b\right) c+b\left(Y^{-}(a, x) c\right).$$
Using the formula in Remark \ref{rmk:lambda}, we see the equivalence.
\end{rmk}
One obtains the Poisson vertex algebra by taking Li's filtration. We recall facts about it here.\par
Let $V$ be a vertex algebra and $$F^{p}V:=\{a_{\left(-n_1-1\right)}^1 \ldots a_{\left(-n_r-1\right)}^r b\mid a^{i}\in V,b\in V,n_{i}\in\mathbb{Z}_{\geq0},n_1+n_2+\cdots+n_r\geq p\}.$$ Then,
\begin{align*}
& V=F^0 V \supset F^1 V \supset \cdots,\ \bigcap_p F^p V=0, \\
& T F^p V \subset F^{p+1} V, \\
& a_{(n)} F^q V \subset F^{p+q-n-1} V \text { for } a \in F^p V, n \in \mathbb{Z}, \\
& a_{(n)} F^q V \subset F^{p+q-n} V \text { for } a \in F^p V, n \geq 0,
\end{align*}
and we set $F^{p}V=V,\ \forall p<0$. The filtration $\{F^{p}V\}$ is called the Li filtration of $V$.\par
Follows {\cite[Proposition 4.2]{Li1}}, let $\grF V:=\bigoplus_{p}F^{p}V/F^{p+1}$ be the associated graded vector space, then $\grF V$ has a Poisson vertex algebra structure given by the following theorem.\par
Let $V$ be a vertex algebra 
the associated graded vector space of $V$ has a Poisson vertex algebra structure defined by
\begin{align}
\sigma_p(a) \cdot \sigma_q(b) & :=\sigma_{p+q}\left(a_{(-1)} b\right)\label{eq:thmLi1}, \\
\partial \sigma_p(a) & :=\sigma_{p+1}(T a)\label{eq:thmLi2}, \\
\sigma_p(a)_{(n)} \sigma_q(b) & :=\sigma_{p+q-n}\left(a_{(n)} b\right)\label{eq:thmLi3},
\end{align}
for all  $a \in F^p V \setminus F^{p+1} V,\ b \in F^q V,\ n \geq 0$.
\begin{rmk}
In terms of generating series, the last property \eqref{eq:thmLi3} can be written as $$Y^{-}(\sigma_{p}(a),z)\sigma_q(b)=\sum_{n\geq0}\sigma_{p+q-n}(a_{(n)}b)z^{-n-1}.$$
\end{rmk}
Let $X$ be a finite type affine scheme i.e., $X=\Spec R$, where $R$ is a finitely generated unital, associative, commutative $\mathbb{C}$-algebra. The $m$-jet scheme of $X$ denoted by $\Jm X$ is characterized by the following property: for every unital
commutative, associative algebra $A$,
$$\Hom_{\mathrm{Sch}}(\mathrm{Spec}\ A,\Jm X)\cong\Hom_{\mathrm{Sch}}(\mathrm{Spec} A[z]/(z^{m+1}),X).$$
If $m>n$, we have a family of projection maps $\pi_{mn}:J_{m}X\rightarrow J_{n}X$. This family of projection maps yields a projective system $\{\pi_{mn},J_m X\}_{m\geq n}$ of schemes. We define the infinite jet scheme $\Jinf X$ of X as $\Jinf X:=\varprojlim J_m X$.
Choose the presentation of algebra $R$ as $R=\mathbb{C}[x^1,x^2,\cdots,x^r]/\langle f_1,f_2,\cdots f_s\rangle$. Then, $J_{m}X$ can be realized as follows. Define variables $x^j_{(-i)},\ i=1,2,\cdots,m+1$ and derivation $T$ on the ring $\mathbb{C}[x^j_{(-i)}\mid i=1,2,\cdots,m+1,j=1,2,\cdots,r]$. The $T$ acts as $$T x_{(-i)}^j= \begin{cases}i x_{(-i-1)}^j & \text { for } i \leq m, \\ 0 & \text { for } i=m+1,\end{cases}.$$ 
Identify $x^j$ and $x^j_{(-1)}$ here, and let 
$$\Jm R=\mathbb{C}\left[x_{(-i)}^j \mid 1\leq i\leq m+1,1\leq j\leq r\right] /\left\langle T^j f_i \mid 1\leq i\leq s,0\leq j\leq m+1\right\rangle.$$
Then, $J_{m}X\cong \mathrm{Spec} J_{m}R$ as affine schemes. The infinite jet scheme of $X$ is realized by taking $m\rightarrow\infty,$ 
$$\Jinf R=\mathbb{C}\left[x_{(-i)}^j \mid i\geq1, j\geq1\right] /\left\langle T^j f_i \mid 1\leq i\leq s,j\geq0\right\rangle,$$
and $\Jinf X=\mathrm{Spec}\Jinf R$.
The following theorem is proved by Arakawa in \cite{Ara2}.
\begin{thm}[{\cite[Proposition 2.3.1]{Ara2}}]\label{thm:Ar}
Let $R$ be a Poisson algebra with Poisson bracket $\{\cdot,\cdot\}$. Then $\Jinf R$ has a unique vertex Poisson algebra structure, called the vertex Poisson structure at level 0, such that
$$u_{(n)} v= \begin{cases}\{u, v\} & \text { if } n=0, \\ 0 & \text { if } n>0,\end{cases}$$
$\forall u,v\in R\subset\Jinf R$.
Let $\mathfrak{g}$ be a Lie algebra and consider the Poisson vertex Lie algebra structure on the $\mathbb{C}[\Jinf\mathfrak{g}^{*}]$. Then by taking the Li's filtration, $\forall k\in\mathbb{C},$ 
$$\mathrm{gr}^{F}V^{k}(\mathfrak{g})\cong\mathbb{C}[\Jinf\mathfrak{g}^{*}],$$
as Poisson vertex algebras where $\mathbb{C}[\Jinf\mathfrak{g}^{*}]$ is equipped with the Poisson vertex algebra structure at level 0.
In particular, at the critical level $k=-h^{\vee}$, we have $\grF\mathfrak{z}(\widehat{\mathfrak{g}})\subset\grF V^{k}(\mathfrak{g})\cong\mathbb{C}[\Jinf\mathfrak{g}^{*}]$ as Poisson vertex algebras.
\end{thm}
\begin{rmk}\label{rmk:level0}
Using the $\lambda$-bracket language, the Poisson vertex algebra structure at level 0 on $\mathbb{C}[\Jinf\mathfrak{g}^{*}]$ can be viewed as follows. Let $L:=v_{0}(\mathfrak{g})$ be an affine vertex Lie algebra at level 0, see Example \ref{exa:AVLA}, then $\forall u,v\in L$ the $\lambda$-bracket is given by $\{u_{\lambda}v\}:=[u,v]$. Extend it to $\Sym(v_k(\FG))$ and consider the modified Poisson vertex algebra $\Sym'(v_0(\FG))$. By the Master Formula \eqref{MF} we give below, then Theorem \ref{thm:elambda} guarantees that $\mathrm{gr}^{F} V^{-h^{\vee}}(\mathfrak{g})\cong \mathrm{Sym}'(v_{0}(\mathfrak{g}))$ as Poisson vertex algebras.
\end{rmk}
\section{Preliminaries II}\label{secpre2}
In the above section, we see a Poisson vertex algebra $\mathscr{V}$ is a differential algebra with a $\lambda$-bracket $\{\cdot_{\lambda}\cdot\}$. In order to discuss Hamiltonian PDEs we have to extend the $\lambda$-bracket to a larger space. We first define what kind of larger space we want.
\begin{dfn}\label{dfn:ADF}
An algebra of differential functions $\mathscr{V}$ in a set of variables $\{u^{i}\}_{i\in I}$ is a differential algebra, which contains the algebra of polynomials $$\mathscr{P}_{\e}:=\mathbb{C}[u^{i,n}\mid i\in\mathrm{I},n\in\mathbb{Z}_{+}],$$
where $u^{i,n}:=\partial^{n}u^{i}$ and $\lvert \mathrm{I} \rvert=\e$ as a differential subalgebra. In addition, $\mathscr{V}$ endows with linear maps $\frac{\partial}{\partial u^{i,n}}:\mathscr{V}\rightarrow\mathscr{V},\ \forall i\in I,n\in\mathbb{Z}_{+}$, which are commuting derivations on $\mathscr{V}$. Extend the derivation $\partial$ in $\mathscr{P}_{\e}$ and such that
\begin{enumerate}
\item For all $f\in\mathscr{V},\ \frac{\partial f}{\partial u^{i,n}}=0$ for all but finitely many pairs $(i,n)\in\mathrm{I}\times\mathbb{Z}_{+}$,
\item $[\frac{\partial}{\partial u^{i,n}},\partial]=\frac{\partial}{\partial u^{i,n-1}}$(Here we set the $n=0$ case, the right-hand side is equal to 0).
\end{enumerate}
\end{dfn}
\begin{exa}\label{exa:adf}
$\mathscr{P}_{\e}$ itself is an example of the algebra of differential functions.
\end{exa}
\subsection{Hamiltonian PDE and PVA}\label{subsec:HPP}
The following Lemma is key for the connection of Hamiltonian PDE and PVA. In particular, in view of the Hamiltonian field theory, the Lemma \ref{Basic Lemma} and Theorem \ref{thm3.1} show that a Poisson vertex algebra $\mathscr{V}$ gives a space of all physical observables for our field-theoretic system.
\begin{lemm}[{\cite[Lemma 8]{K1}}]\label{Basic Lemma}
Let $\mathscr{V}$ be a PVA. Let $\overline{\mathscr{V}}:=\mathscr{V}/\partial\mathscr{V}$ and let $\int:\mathscr{V}\rightarrow\overline{\mathscr{V}}$ be the quotient map. Then we have the following well-defined brackets:
\begin{enumerate}
\item $\overline{\mathscr{V}}\times\overline{\mathscr{V}}\rightarrow\overline{\mathscr{V}},\ \{\int a,\int b\}:=\int\{a_{\lambda}b\}_{\lambda=0},$
\item $\overline{\mathscr{V}}\times\mathscr{V}\rightarrow\overline{\mathscr{V}},\ \{\int a,b\}:=\{a_{\lambda}b\}_{\lambda=0}.$
\end{enumerate}
Moreover, $(1)$ defines a Lie algebra on $\overline{\mathscr{V}}$, and $(2)$ defines a representation of the Lie algebra $\overline{\mathscr{V}}$ on $\mathscr{V}$ and the $\lambda$-bracket of $\mathscr{V}$ commuting with $\partial$. (Here we first set $\lambda=0$, then take the quotient map.)
\end{lemm}
The following theorem extends the $\lambda$-bracket on a Poisson vertex algebra to an algebra of differential functions $\mathscr{V}$ that contains it.
\begin{thm}[{\cite[Theorem 8]{K1}}]\label{thm:elambda}
Let $\mathscr{V}$ be an algebra of differential functions in the variables $\{u^{i}\}_{i{\in}\mathrm{I}}$. For each pair $i,j\in\mathrm{I}$ choose $\{u^{i}_{\lambda}u^{j}\}=H_{ji}(\lambda)\in\mathscr{V}[\lambda]$, then
\begin{enumerate}
\item The Master formula:
\begin{equation}\label{MF}
\{f_{\lambda}g\}=\sum_{i,j\in I,p,q\in\mathbb{Z}_{+}}\frac{\partial g}{\partial u^{j,q}}\{u^{i}_{\partial+\lambda}u^{j}\}_{\rightarrow}(-\partial-\lambda)^{p}\frac{\partial f}{\partial u^{i,p}},
\end{equation}
defines a $\lambda$-bracket on $\mathscr{V}$, which satisfies sesquilinearity, the left, and right Leibniz rules, and extends the given $\lambda$-bracket on the variables $u_{i}$'s. Consequently, any $\lambda$-bracket on the algebra of differential polynomials satisfying these properties is given by the Master Formula.
\item This $\lambda$-bracket is skewsymmetric if skew-symmetry holds for every pair of variables:
\begin{equation*}
\{u^{i}_{\lambda}u^{j}\}=-\{u^{j}_{\lambda}u^{i}\},\ \forall i,j\in\mathrm{I},
\end{equation*}
\item 
If this $\lambda$-bracket is skewsymmetric, then it satisfies the Jacobi identity, provided Jacobi identity holds for every triple of variables:
\begin{equation*}
\{u^{i}_{\lambda}\{u^{j}_{\mu}u^{k}\}\}-\{u^{j}_{\mu}\{u^{i}_{\lambda}u^{k}\}\}=\{\{u^{i}_{\lambda}u^{k}\}_{\lambda+\mu}u^{j}\},\ \forall i,j,k\in\mathrm{I}. 
\end{equation*}
\end{enumerate}
\end{thm}
\begin{dfn}\label{dfn:PS}
Given a Poisson vertex algebra $\mathscr{V}$ with generators $\{u^i\}_{i\in \mathrm{I}},\ \lvert\mathrm{I}\rvert=\e$ and $\lambda$-bracket $\{\cdot_{\lambda}\cdot\}$. Define a matrix differential operator $H(\partial)\in Mat_{\e\times\e}\mathscr{V}[\partial]$ with entries given by $H_{ij}(\partial):=\left\{u^{j} _{\partial+\lambda} u^i\right\}_{\rightarrow}\rvert_{\lambda=0}$. We will call $H(\partial)$ the Poisson structure corresponding to the Poisson vertex algebra $\mathscr{V}$.
\end{dfn}
\begin{dfn}\label{dfn:HEQ}
Given a PVA $\mathscr{V}$ and a local functional $\int h\in\overline{\mathscr{V}}$, the associated Hamiltonian PDE is 
\begin{equation*}
\frac{d u}{dt}=\{\int h,u\}.	
\end{equation*}
The local functional $\int h$ is called the Hamiltonian of this equation.
\end{dfn}
\begin{thm}\label{thm3.1}
Let the PVA $\mathscr{V}$ be an algebra of differential functions in the variables $\{u_{i}\}_{i\in I}$ and the $\lambda$-bracket is given by the Master Formula \eqref{MF}, then
\begin{enumerate}
\item Hamiltonian PDE: $\frac{d u}{dt}=\{\int h,u\}=H\frac{\delta}{\delta u}\int h$;
\item Poisson bracket on $\mathscr{V}$: $\{\int f,\int g\}=\int\frac{\delta g}{\delta u}\cdot H\frac{\delta f}{\delta u}$.	
\end{enumerate}
where $H$ is a matrix differential operator that we defined in Definition \ref{dfn:PS}.
\end{thm}
Through the above discussions, we re-formulate Liouville integrability of a field-theoretic system in a rigorous way.
\begin{dfn}\label{dfn:LI}
A Hamiltonian system $u_t(x,t)=\{\int h,u\}$. We say this system is Liouville integrable if $\int h$ is contained in an infinite-dimensional Poisson commutative
subalgebra in $\bar{\mathscr{V}}$ with respect to $\{\cdot,\cdot\}$.
\end{dfn}
\begin{exa}\label{exa:KdV}
Let $\mathscr{V}:=\mathbb{C}[u,u',u'',\cdots]$ be a Poisson vertex algebra with $\lambda$-bracket $\{u_{\lambda}u\}=\lambda$. Take $h=\int \frac{1}{2} u_x^2+u^3 d x,$ then $\frac{d u}{d t}=6uu_x+u_{xxx}$. The Poisson vertex algebra here we use can be obtained by taking the classical limit of Heisenberg vertex algebra \cite{K1}. The integrability of the KdV equation can be argued by its bi-Hamiltonian structure, which was first discovered by Faddeev and Zakharov, see \cite{FZ}.
\end{exa}
\section{Modified Yang-Baxter equations of vertex Lie algebra}\label{sec:MYBE}
\begin{dfn}\label{dfn:rmatrix}
Let $(L,Y^{-},T)$ be a vertex Lie algebra. A linear map $R_{ L}\in \End L$ is called an $R$-matrix for the vertex Lie algebra, if 
\begin{enumerate}
\item[1.] $(L,Y^{-}_{R_L}(?,z),T)$ is a vertex Lie algebra where  
\begin{equation*}\label{eq:rmatrix}
Y^{-}_{R_{L}}(u,z):=Y^{-}(R_{L}(u),z)+Y^{-}(u,z)R_{L},	
\end{equation*}
\item[2.] $R_L$ satisfies
\begin{equation}\label{eq:4.1}
[R_{L},T]=0.	
\end{equation}
\end{enumerate}
We will denote the new vertex Lie algebra simply by $L_{R_{L}}$ and call it the $R_{L}$-twisted vertex Lie algebra.
\end{dfn}
\begin{thm}
Let $R_{L}\in \mathrm{End} L$ satisfy \eqref{eq:4.1}, then, each of the following conditions is sufficient for $R_{L}$ to be an $R$-matrix,
\begin{align}
&Y^{-}(R_{L}(u),z)R_{L}=R_{L}Y^{-}_{R_{L}}(u,z)\ \forall u\in L\label{CCYB},\\
&Y^{-}(R_{L}(u),z)R_{L}-R_{L}Y^{-}_{R_{L}}(u,z)=-Y^{-}(u,z)\ \forall u\in L\label{MCYB}.	
\end{align}
The equation \eqref{CCYB} is called the constant classical Yang-Baxter equation for the vertex Lie algebra L, and the equation \eqref{MCYB} is called the modified Yang-Baxter equation $(\mathrm{mYBE})$ for the vertex Lie algebra L.
\end{thm}
\begin{proof}
In order to show this theorem, we only need to find sufficient conditions for $Y^{-}_{R_{L}}$ to satisfy the Yang-Baxter equation.
For convenience, we introduce the following notations:
\begin{equation*}
A:=x^{-1}\delta(\frac{y-z}{x}),\ B:=x^{-1}\delta(\frac{z-y}{-x}),\ C:=z^{-1}\delta(\frac{y-x}{z}).
\end{equation*}	
Under this notation, the half Jacobi identity of vertex Lie algebra $(L, Y^{-}, T)$ reads
\begin{equation*}
	(AY^{-1}(u,y)Y^{-}(v,z)-BY^{-}(v,z)Y^{-}(u,y))_{-}=(CY^{-}(Y^{-}(u,x)v,z))_{-},\ \forall u,v\in L,
\end{equation*}
and we have
\begin{equation}\label{eq:4.4} 
\begin{aligned}
\lefteqn{(AY^{-}_{R_{L}}(u,y)Y^{-}_{R_{L}}(v,z)-BY^{-}_{R_{L}}(v,z)Y^{-}_{R_{L}}(u,y))_{-}}\\
&=(AY^{-}(R_{L}(u),y)Y^{-}(R_{L}(v),z)-BY^{-}(R_{L}(v),z)Y^{-}(R_{L}(u),y))_{-}\\
   &+(A[Y^{-}(R_{L}(u),y)Y^{-}(v,z)R_{L}+Y^{-}(u,y)R_{L}Y^{-}(R_{L}(v),z)+Y^{-}(u,y)R_{L}Y^{-}(v,z)R_{L}])_{-}\\
   &-(B[Y^{-}(R_{L}(v),z)Y^{-}(u,y)R_{L}+Y^{-}(v,z)R_{L}Y^{-}(u,x)R_{L}+Y^{-}(v,z)R_{L}Y^{-}(R_{L}(u),y)])_{-}\\
   &{=}(CY^{-}(Y^{-}(R_{L}(u),x)v),z)_{-}+(A[Y^{-}(R_{L}(u),y)Y^{-}(v,z)R_{L}+Y^{-}(u,y)R_{L}Y^{-}(R_{L}(v),z)_{-}\\
   &+Y^{-}(u,y)R_{L}Y^{-}(v,z)R_{L}])_{-}-(B[Y^{-}(R_{L}(v),z)Y^{-}(u,y)R_{L}+Y^{-}(v,z)R_{L}Y^{-}(u,x)R_{L}\\
   &+Y^{-}(v,z)R_{L}Y^{-}(R_{L}(u),y)])_{-},\qquad{(\hbox{by\ the\ half-Jacobi\ identity\ of\ $(L, Y^{-}, T)$}}.
\end{aligned}
\end{equation}
In addition,
\begin{equation}\label{eq:4.5}
\begin{aligned}
\lefteqn{(CY^{-}_{R_{L}}(Y^{-}_{R_{L}}(u,x)v,z))_{-}=(CY^{-}_{R_{L}}(Y^{-}(R_{L}(u),x)+Y^{-}(u,x)R_{L}v,z))_{-}}\\
&=(C[R_{L}Y^{-}(R_{L}(u),x)v,z)+Y^{-}(R_{L}(u),x)v,z)R_{L}\\&+Y^{-}(R_{L}Y^{-}(u,x)R_{L}(v),z)+Y^{-}(Y^{-}(u,x)R_{L}(v),z)R_{L}])_{-}.
\end{aligned}
\end{equation}
Thanks to the half Jacobi identity of $(L, Y^{-}, T),$ we have the following three identities:
\begin{equation}\label{eq:4.6}
\left\{
{
\begin{aligned}
\lefteqn{(CY^{-}(Y^{-}(R_{L}(u),x)R_{L}(v),z))_{-}}\\&=(AY^{-}(R_{L}(u),y)Y^{-}(R_{L}(v),z))_{-}-(BY^{-}(R_{L}(v),z)Y^{-}(R_{L}(u),z))_{-},\\
\lefteqn{(CY^{-}(Y^{-}(R_{L}(u),x)v,z)R_{L})_{-}}\\&=(AY^{-}(R_{L}(u),y)Y^{-}(v,z)R_{L})_{-}-(BY^{-}(v,z)Y^{-}(R_{L}(u),y)R_{L})_{-},\\
\lefteqn{(CY^{-}(Y^{-}(u,x)R_{L}(v),z)R_{L})_{-}}\\&=(AY^{-}(u,y)Y^{-}(R_{L}(v),z)R_{L})_{-}-(BY^{-}(R_{L}(v),z)Y^{-}(u,y)R_{L})_{-}.	
\end{aligned}
}
\right.
\end{equation}
Therefore, letting $Y^{-,*}_{R_{L}}(u,z):=Y^{-}(R_{L}(u),z)R_{L}-R_{L}Y^{-}_{R_{L}}(u,z)\ \forall u\in L,$ and regarding \eqref{eq:4.6}, we see \eqref{eq:4.4}-\eqref{eq:4.5} equals to
\begin{equation}\label{eq:4.9}
(CY^{-}(Y^{-,*}_{R_{L}}(u,x)v,z)-AY^{-}(u,y)Y_{R_{L}}^{-,*}(v,z)+BY^{-}(v,z)Y_{R_{L}}^{-,*}(u,y))_{-}.
\end{equation}
Note that by definition of the half Jacobi identity, the half Jacobi identity holds for $(L,Y^{-}_{R_{L}},T)$ if and only if \eqref{eq:4.9}=0. Hence, if $Y^{-,*}_{R_{L}}(?,z)=0,$ the operator $Y_{R_{L}}^{-}(?,z)$ satisfies the half Jacobi identity. Also, note that if $Y^{-,*}_{R_{L}}(u,z)=aY^{-,*}(u,z),\ \forall u\in L$ and some $a\in\mathbb{C}^{\times},$ by the half Jacobi identity of vertex algebra $(L,Y^{-},T)$, we have \eqref{eq:4.9} equals to
\begin{equation}\label{eq:4.10}
\begin{aligned}
a((AY^{-}(u,y)Y^{-}(v,z))_{-}-(BY^{-}(v,z)Y^{-}(u,y))_{-}-(CY^{-}(Y^{-}(u,x)v,z))_{-})=0
\end{aligned}
\end{equation}
$\forall u,v\in L$. Hence, if $Y^{-,*}_{R_{L}}(u,z)=aY^{-,*}(u,z),\ \forall u\in L$ and some $a\in\mathbb{C}^{\times},$ the half Jacobi identity holds for the vertex Lie algebra $(L,Y^{-,*}_{R_{L}},T)$. By re-scaling the operator $R_{L}$, i.e., replacing $R_{L}$ by $\sqrt{a}R_{L}$, one translates \eqref{eq:4.10} to 
\begin{equation*}
Y^{-,*}_{R_{L}}(?,z)=-Y^{-}(?,z).	
\end{equation*}
Therefore, we conclude that 
\begin{enumerate}
\item $Y^{-}(R_{L}(u),z)R_{L}=R_{L}Y^{-}_{R_{L}}(u,z),\ \forall u\in L$,\\
\item $Y^{-}(R_{L}(u),z)R_{L}-R_{L}Y^{-}_{R_{L}}(u,z)=-Y^{-}(u,z),\ \forall u\in L,$	
\end{enumerate}
are two sufficient conditions for $R_{L}$ to be an $R$-matrix for vertex Lie algebra.
\end{proof}
\section{Factorization theorems}\label{sec:PFT}
We first fix some notations. In this section we consider a vertex Lie algebra $(L, Y^{-}, T),$ and let $R_{L}$ be an $R$-matrix of $L$ satisfying the modified classical Yang-Baxter equation \eqref{MCYB}. Set
\begin{equation}\label{Eq:5.1}
L_{\pm}:=(R_{L}\pm1)(L),\hspace{1cm} \mathscr{L}_{\pm}:=\ker(R_{L}\mp1).	
\end{equation}
\subsection{Factorization theorem A}\label{subsec:FT}
In this section, we state and prove factorization theorem A.
\begin{lemm}\label{lem5.1}
$(L_{\pm},Y^{-},T)$ are vertex Lie subalgebras of $(L,Y^{-},T)$.
\end{lemm}
\begin{proof}
Note that because of equation \eqref{MCYB}, $\forall u\in L$ we have, 
\begin{align*}
(R_{L}\pm1)(Y_{R_{L}}^{-}(u,z))&=R_{L}Y_{R_{L}}^{-}(u,z)\pm Y_{R_{L}}^{-}(u,z)\\
&=Y^{-}(R_{L}(u),z)R_{L}+Y^{-}(u,z)\pm(Y^{-}(R_{L}(u),z)+Y^{-}(u,z)R_{L})\\
&=Y^{-}((R_{L}\pm1)u,z)(R_{L}\pm1).
\end{align*}
This shows that $(L_{\pm},Y^{-},T)$ are vertex Lie subalgebras of $(L,Y^{-},T)$.
Hence, we complete the proof.
\end{proof}
Note that Lemma \ref{lem5.1} shows that the mappings $R_{L}\pm1$ are vertex Lie algebra homomorphisms from $(L,Y^{-}(?,z),T)$ to $(L_{\pm},Y^{-}(?,z),T)$. According to Lemma \ref{lem5.1}, we obtain the following corollary.
\begin{cor}
The subspaces $\mathscr{L}_{\pm}$ are vertex Lie algebra ideals of $(L,Y^{-}_{R_{L}},T)$ and $L_{\pm}\cong L/\mathscr{L}_{\mp}$ as vertex Lie algebras.
\end{cor}
\begin{lemm}\label{lem5.3}
The subspaces $\mathscr{L}_{\pm}$ are vertex Lie algebra ideals of $(L_{\pm},Y^{-},T)$.
\end{lemm}
\begin{proof}
First, we verify that $\mathscr{L}_{\pm}\subset L_{\pm}.$ This follows from $\forall u\in\mathscr{L}_{\pm}, (R_{L}\mp1)u=0,$ hence $R_{L}(u)=\pm u, and\ \pm u=\frac{1}{2}(R_{L}\pm1)(u)\in L_{\pm}.$\\
Second, we show $\forall v\in L,$ the identity $(R_{L}\mp1)Y^{-}((R_{L}\pm1)v,z)u=0,\ \forall u\in\mathscr{L}_{\pm}$ holds. This follows from
\begin{align}
\lefteqn{(R_{L}\mp1)Y^{-}((R_{L}\pm1)v,z)u}\nonumber\\&=R_{L}Y^{-}(R_{L}(v),z)u\pm R_{L}Y^{-}(v,z)u\mp Y^{-}(R_{L}(v),z)u-Y^{-}(v,z)u\\
&=R_{L}Y^{-}(R_{L}(v),z)u-Y^{-}(v,z)u+R_{L}Y^{-}(v,z)R_{L}(u)-Y^{-}(R_{L}(v),z)R_{L}(u)\nonumber\\
&=[R_{L}Y^{-}_{R_{L}}(v,z)-Y^{-}(R_{L}(v),z)R_{L}-Y^{-}(v,z)]u\nonumber\\
&=0,	\nonumber
\end{align}
because of the modified Yang-Baxter equation \eqref{MCYB} and $R_{L}(u)=\pm u,\ \forall u\in\mathscr{L}_{\pm}$. The inclusion $T\mathscr{L}_{\pm}\subset\mathscr{L}_{\pm}$ follows from $[T,R_{L}]=0.$ Thus, we finished the proof.
\end{proof}
Let $\overline{\phantom{L}}:L_{\pm}\longrightarrow L_{\pm}/\mathscr{L}_{\pm}$ be the quotient map, and $\forall u,v\in L_{\pm}$, we define
\begin{align*}
&\overline{T}\overline{u}:=\overline{Tu},\ \overline{Y}^{-}(\overline{u},z)v:=\sum_{n\geq0}\overline{u_{n}v}z^{-n-1}.
\end{align*}
Note that by Lemma \ref{lem5.3}, we see that $(L_{\pm}/\mathscr{L}_{\pm},\overline{T},\overline{Y}^{-})$	are vertex Lie algebras.
\begin{lemm}
The map
\begin{equation}\label{Eq:5.7}
\theta(\overline{(R_{L}+1)u}):=\overline{(R_{L}-1)u},\ \forall u\in L,	
\end{equation}
is a vertex Lie algebra isomorphism from $(L_{+}/\mathscr{L}_{+},\overline{Y}^{-},\overline{T})$ to $(L_{-}/\mathscr{L}_{-},\overline{Y}^{-},\overline{T})$.
\end{lemm}
\begin{proof}
The well-definedness directly follows from \eqref{Eq:5.1}. For the injectivity, we have
\begin{equation*}
\begin{aligned}
\ker\theta&=\{(R_{L}+1)u\mod\mathscr{L}_{+}\mid u\in L,\  \overline{(R_{L}-1)u}=0\}\\
&=\{(R_{L}+1)u\mod\mathscr{L}_{+}\mid u\in L,\ (R_{L}-1)u\in\mathscr{L}_{-}=\ker(R_{L}+1)\}\\
&=\{(R_{L}+1)u\mod\mathscr{L}_{+}\mid u\in L,\ (R_{L}+1)(R_{L}-1)u=0\}\\
&=\{(R_{L}+1)u\mod\mathscr{L}_{+}\mid u\in L,\ (R_{L}-1)(R_{L}+1)u=0\}\\
&=\{0\}.
\end{aligned}
\end{equation*}
By definition, the mapping $\theta$ is also surjective, and so it is invertible. Now we verify that $[\overline{T},\theta]=0$. Because of $[T, R_{L}]=0$ and definition of $\theta$, we have
\begin{equation*}
\begin{aligned}
\lefteqn{\theta(\overline{T}\overline{(R_{L}+1)u})}\\&=\theta(\overline{T(R_{L}+1)u})\\
&=\theta(\overline{(R_{L}+1)Tu})\\
&=\overline{(R_{L}-1)Tu}\\
&=\overline{T}\theta(\overline{(R_{L}+1)u}),
\end{aligned}
\end{equation*}
$\forall u\in L$. On the other hand,
\begin{align*}
\lefteqn{\theta(\bar{Y}^{-}(\overline{(R_{L}+1)u},z)((R_{L}+1)v))}\\&=\theta(\overline{Y^{-}((R_{L}+1)u,z)(R_{L}+1)v})\\
&=\theta(\overline{Y^{-}(R_{L}u,z)R_{L}v+Y^{-}(u,z)R_{L}v+Y^{-}(u,z)v+Y^{-}(R_{L}(u),z)v})
\\
&=\theta(\overline{R_{L}Y^{-}_{R_{L}}(u,z)v+Y^{-}_{R_{L}}(u,z)v}) \qquad(\hbox{by\ mYBE \eqref{eq:MYBE}})\\
&=\theta(\overline{(R_{L}+1)Y^{-}_{R_{L}}(u,z)v})\\
&=\theta(\overline{(R_{L}+1)Y_{R_{L}}^{-}(u,z)v})\\
&=\overline{(Y^{-}(R_{L}u,z)R_{L}v+Y^{-}(u,z)v-Y^{-}(R_{L}u,z)v-Y^{-}(u,z)R_{L}v)}\qquad{(\hbox{by\ mYBE\ \eqref{eq:MYBE}})}\\
&=\overline{Y^{-}((R_{L}-1)u,z)(R_{L}-1)v}\\
&=\bar{Y}^{-}(\overline{(R_{L}-1)u},z)\overline{(R_{L}-1)v}\\
&=\bar{Y}^{-}(\theta(\overline{(R_{L}+1)u)},z)\theta(\overline{(R_{L}+1)v}).
\end{align*}
Therefore, we complete the proof.
\end{proof}
\begin{FTA}\label{FTA}
Consider a vertex Lie algebra $L$ and let $R_{L}$ be an $R$-matrix of L. Let $\nu:L\longrightarrow L_{+}\oplus L_{-}$ defined by 
\begin{equation*}
\nu(u):=((R_{L}+1)u,(R_{L}-1)u),\ \forall u\in L,	
\end{equation*}
and $\tau:L_{+}\oplus L_{-}\longrightarrow L$ by
\begin{equation*}
	\tau(u_{+},u_{-}):=u_{+}-u_{-},\ \forall u_{+}\in L_{+},\ u_{-}\in L_{-}.
\end{equation*}
\begin{enumerate}
\item[(i)] The map $\nu$ is an injective vertex Lie algebra homomorphism from $(L,Y^{-}_{R_{L}},T)$ to $(L_{+}\oplus L_{-},Y^{-}\oplus Y^{-},T\oplus T)$ and
\begin{equation*}
\nu(L)=\{(u_{+},u_{-})\in L_{+}\oplus L_{-}\mid\theta(\overline{u_{+}})=\overline{u_{-}}\}.
\end{equation*}
\item[(ii)] Every element $v\in L$ can be uniquely written as $v=v_{+}-v_{-}$ such that $(v_{+},v_{-})\in\nu(L)$.
\end{enumerate}
\end{FTA}
\begin{proof}[Proof of factorization theorem A]
(i) follows from Lemma \ref{lem5.1}--Lemma \ref{lem5.3}. We show (ii) here. $\forall v\in L,$ we can write
\begin{equation*}
v=(R_{L}+1)(\frac{v}{2})-(R_{L}-1)(\frac{v}{2})\quad\hbox{and}\quad((R_{L}+1)\frac{v}{2},(R_{L}-1)\frac{v}{2})\in\nu(L).
\end{equation*}
To prove uniqueness, suppose that $v=v_{+}-v_{-}$ with $(v_{+},v_{-})\in\nu(L),$ then $\exists w\in L,$ such that $v_{+}=(R_{L}+1)(w),\ v_{-}=(R_{L}-1)w$. Note the element $w$ satisfies 
\begin{equation*}
	w=\frac{1}{2}[(R_{L}+1)w-(R_{L}-1)w)]=\frac{1}{2}(v_{+}-v_{-})=\frac{v}{2}.
\end{equation*}
This shows the uniqueness of the decomposition we claimed in the factorization theorem A.
\end{proof}
\subsection{Factorization theorem B}
In this section, we state and prove factorization theorem B.
\begin{FTB}\label{FTB}
Let $(L,Y^{-},T)$ be a vertex Lie algebra and let $(L_{+},Y^{-},T)\ and\ \\(L_{-},Y^{-},T)$ be two vertex Lie subalgebras. Also,
suppose $\widetilde{L}$ is a vertex Lie subalgebra of $(L_{+}\oplus L_{-},Y^{-}\oplus Y^{-},T\oplus T).$
\begin{enumerate}
\item[(i)] Suppose that every vector $u\in L$ can be uniquely expressed as
\begin{equation*}
u=u_{+}-u_{-}\ with\ (u_{+},u_{-})\in\widetilde{L}.
\end{equation*}
Define $R_{L}\in\End(L)$ by
\begin{equation*}
R_{L}(u)=u_{+}+u_{-}\ \mathit{for}\ u=u_{+}-u_{-},\ (u_{+},u_{-})\in\widetilde{L}.	
\end{equation*}
Then $R_{L}$ satisfies the modified Yang-Baxeter equation for vertex Lie algebra, i.e., $R_L$ satisfies
\begin{equation}\label{eq:MYBE}
Y^{-}(R_{L}(u),z)R_{L}-R_{L}Y^{-}_{R_{L}}(u,z)=-Y^{-}(u,z),\ \forall u\in L,
\end{equation}
 and $[R_{L},T]=0$.
\item[(ii)] Suppose, in addition, there are vertex Lie algebra ideals $\mathscr{L}_{\pm}\subset L_{\pm}$ (w.r.t. the vertex Lie algebra structures ($L_{\pm}$,$Y^{-}(?,z)$, T)) and a vertex Lie algebra isomorphism $\theta:L_{+}/\mathscr{L}_{+}\rightarrow L_{-}/\mathscr{L}_{-}$ such that
\begin{equation*}
\widetilde{L}=\{(u_{+},u_{-})\in L_{+}\oplus L_{-}\mid \theta(\overline{u_{+}})=\overline{u_{-}}\}.	
\end{equation*}
Then,
\begin{equation*}
L_{\pm}=(R_{L}\pm1)L,\ \mathscr{L}_{\pm}=\ker(R_{L}\mp1).
\end{equation*}
\end{enumerate}
In particular, if $L_{1},\ L_{2}$ are vertex Lie subalgebras of $L$ and $L=L_{1}\oplus L_{2},$ the operator $R_{L}$ defined by
\begin{equation*}
R_{L}(u)=u_{1}+u_{2}\ for\ u=u_{1}-u_{2}\in L,\ u_{i}\in L_{i}.	
\end{equation*}
satisfies the modified Yang-Baxter equation.
\end{FTB}
\begin{proof}[Proof of factorization theorem B]
For the first part, we have
\begin{align*}
R_{L}T(u)&=R_{L}(T(u_{+})-T(u_{-}))\\
&=T(u_{+})+T(u_{-})\qquad{(\hbox{Because of $Tu_{+}\in L_{+},\ Tu_{-}\in L_{-}$})}\\
&=T(u_{+}+u_{-})\nonumber\\
&=TR_{L}(u).\nonumber	
\end{align*}
This shows that $[T,R_{L}]=0$.
Also $\forall u,v\in L$
\begin{align*}
	Y^{-}_{R_{L}}(u,z)v&=Y^{-}(u_{+}-u_{-},z)(v_{+}+v_{-})+Y^{-}(u_{+}+u_{-},z)(v_{+}-v_{-})\\
	&=Y^{-}(u_{+},z)v_{+}+Y^{-}(u_{+},z)v_{-}-Y^{-}(u_{-},z)v_{+}-Y^{-}(u_{-},z)v_{-}\nonumber\\
	&\quad{}+Y^{-}(u_{+},z)v_{+}-Y^{-}(u_{+},z)v_{-}+Y^{-}(u_{-},z)v_{+}-Y^{-}(u_{-},z)v_{-}\nonumber\\
	&=2(Y^{-}(u_{+},z)u_{+}-Y^{-}(u_{-},z)v_{-}).\nonumber
\end{align*}
We verify the modified Yang-Baxter equation of $R_{L}$ now. Note that we have
\begin{align*}
Y&^{-}(R_{L}u,z)R_{L}(v)-R_{L}Y^{-}_{R_{L}}(u,z)v\\
&=Y^{-}(u_{+}+u_{-},z)(v_{+}+v_{-})-R_{L}(2Y^{-}(u_{+},z)v_{+}-2Y^{-}(u_{-},z)v_{-})\nonumber\\
&=Y^{-}(u_{+},z)v_{+}+Y^{-}(u_{+},z)v_{-}+Y^{-}(u_{-},z)v_{+}+Y^{-}(u_{-},z)v_{-}\nonumber\\
&\quad{}-2Y^{-}(u_{+},z)v_{+}-2Y^{-}(u_{-},z)v_{-}\nonumber\\
&=-Y^{-}(u_{+},z)(v_{+}-v_{-})+Y^{-}(u_{-},z)(v_{+}-v_{-})\nonumber\\
&=-Y^{-}(u_{+},z)(v)+Y^{-}(u_{-},z)v\nonumber\\
&=-Y^{-}(u,z)v.\nonumber
\end{align*}
Now, we prove the second statement. Since the statement is symmetry with respect to the signs $+$ and $-$, it is sufficient to show:
\begin{equation*}
	L_{+}=(R_{L}+1)(L),\hspace{1cm} \mathscr{L}_{+}=\ker(R_{L}-1).
\end{equation*} 
In this case, $\forall u\in L_{+},\ \exists v\in L_{-}$ such that $\theta(\bar{u})=\bar{v}.$ Take $w:=\frac{u-v}{2},$ we get $u=(R_{L}+1)(w)\in(R_{L}+1)(L).$ On the other hand, $\forall u\in L$ 
\begin{equation*}
(R_{L}+1)(u)=(u_{+}+u_{-})+(u_{+}-u_{-})=2u_{+}.	
\end{equation*}
Hence, we have proved $L_{+}=(R_{L}+1)(L)$. An analogous argument shows that $L_{-}=(R_{L}-1)(L)$. For the second equality, since $L_+=(R_L+1)(L)$, we have $\forall u\in\mathcal{L}_+\subset L_+=(R_L+1)L.$ Therefore, $\exists\ x\in L$ such that $(R_L+1)(x)=u.$ In addition, we have,
\begin{align*}
(R_L-1)u&=(R_L-1)(R_L+1)(x)\\
&=(R_L-1)(R_L+1)(x_+-x_-)\\
&=(R_L-1)(x_++x_-+x_+-x_-)\\
&=(R_L-1)(2x_+)\\
&=0.
\end{align*}
To show the reverse inclusion, we note that 
\begin{align*}
\ker(R_L-1)&=\{u\in L\mid (R_L-1)u=0\}\\
&=\{u\in L\mid u_++u_--u_++u_-=0\}\\
&=\{u\in L\mid 2u_-=0\}.
\end{align*}
Hence, $\forall u\in\ker(R_L-1),\ u=u_+-u_-=u_+$ and
\[\theta(u_+\ \mathrm{mod}\ \mathcal{L}_+)=u_-\ \mathrm{mod}\ \mathcal{L}_-=0.\]
Because $\theta$ is an isomorphism, we conclude that $u_+\in\mathcal{L}_+.$
\end{proof}
As a corollary, we get:
\begin{cor}\label{cor5.1}
Let $L_{1}$ and $L_{2}$ be two vertex Lie subalgebras of $L$ and $L=L_{1}\oplus L_{2},$ then the operator $R_{L}\in\End L$, defined by 
\begin{equation*}
R_{L}(u)=u_{1}-u_{2},\ for\ u=u_{1}+u_{2},\ u_{i}\in L,
\end{equation*}
satisfies the conditions of an $R$-matrix.
\end{cor}
\subsection{Comparison of the classical $R$-matrices of Lie algebras and VLA}\label{subsec:CRLCRB} In this subsection, we indicate the connection between the classical $R$-matrix of Lie algebras and the classical $R$-matrix of vertex Lie algebras. For convenience we use the $\lambda$-bracket notation, in particular, we introduce the following definition. 
\begin{dfn}
Let $L$ be a vertex Lie algebra with $\lambda$-bracket $\{\cdot_{\lambda}\cdot\}$, and $R_L$ be a classical $R$-matrix of vertex Lie algebra $L$. Based on the above discussions, the $R$-matrix $R_{L}$ gives another $\lambda$-bracket structure on the vector space $L$, we will denote this new $\lambda$-bracket by $\{\cdot_{\lambda}\cdot\}_{R_{L}}$ and call it the $\lambda$-bracket twisted by the $R$-matrix $R_{L}$.
\end{dfn}
\begin{rmk}\label{rmk:eqrmatrix}
Use the $\lambda$-bracket notation, Definition \ref{dfn:rmatrix} is equivalent to the following: Let $(L,\{\cdot_{\lambda}\cdot\},T)$ be a vertex Lie algebra. A linear map $R_{ L}\in \mathrm{End} L$ is called an $R$-matrix for the vertex Lie algebra, if the $\lambda$-bracket defined by
\begin{equation*}
\{u_{\lambda}v\}_{R_{L}}=\{R_{L}(u)_{\lambda}v\}+\{u_{\lambda}R_{L}(v)\},
\end{equation*}
defines vertex Lie algebra structure on $L$, and if
\begin{equation*}
[R_{L},T]=0.	
\end{equation*}
\end{rmk}
Let $\mathfrak{g}$ be a finite-dimensional Lie algebra with dimension $\e$ and let $u^{1},u^{2},\cdots,u^{\e}$ be its basis. Consider its affine vertex Lie algebra $v_0(\FG)$ at level $0$ with $\lambda$-bracket $\{\cdot_{\lambda}\cdot\}$ and its vertex Lie subalgebra
\[v_0(\FG)_0=\Span_{\C}\{u^{i}\mid i=1,2,\cdots,\dim\FG\},\]
then $v_0(\FG)_{0}\cong\FG$ as Lie algebras. In addition, the Lie bracket on the Lie algebra $v_0(\FG)_0$ coincides with the $\lambda$-bracket $\{\cdot_{\lambda}\cdot\}$. Let $R_L^0\in\End v_0(\FG)$ be a classical $R$-matrix of $v_0(\FG)_0$, then by the equivalent definition that we give in Remark \ref{rmk:eqrmatrix}, we see $R_{L}^0$ is a classical $R$-matrix of Lie algebra $v_0(\FG)_0\cong\FG$. In addition, we see that using the $\lambda$-bracket notation, if $R_L\in\End v_0(\FG)$ the classical and modified Yang-Baxter equations of vertex Lie algebra can be rewritten as $\forall u,v\in v_0(\FG),$
\begin{align*}
&\{R_L(u)_{\lambda}R_L(v)\}=R_L\{u_{\lambda}v\},\\
&\{R_{L}(u)_{\lambda}R_{L}(v)\}-R_{L}\{u_{\lambda}v\}_{R_{L}}=-\{u_{\lambda}v\}.
\end{align*}
Hence, the classical and modified Yang-Baxter equations of the vertex Lie subalgebra $v_0(\FG)_0$ are the same as the classical and modified Yang-Baxter equations for the Lie algebra $\FG$.
\subsection{Examples of the classical $R$-matrix}\label{subsec:ECR}
In this subsection, we give two examples of the classical $R$-matrix for the $v_{k}(\mathfrak{g})$ given by the universal affine vertex algebras $V^{k}(\mathfrak{g})$. In particular, we calculate $\lambda$-brackets $\{\cdot_{\lambda}\cdot\}_{R_{L}}$ on the generators for $v_{k}(\mathfrak{g})_{R_{L}}$ obtained from the Borel and Iwasawa decompositions of $\mathfrak{g}$. Regarding Corollary \ref{cor5.1}, we have the following theorem, which gives a convenient way to construct the classical $R$-matrix for $v_{k}(\mathfrak{g})$.\par
Let $\mathfrak{g}$ be a finite-dimensional Lie algebra that admits a Lie algebra decomposition $\mathfrak{g}\cong\mathfrak{a}\oplus\mathfrak{b}$ i.e., $\FG\cong\mathfrak{a}\oplus\mathfrak{b}$ as vector space and $\mathfrak{a}$, $\mathfrak{b}$ are Lie subalgebras of $\FG$. In addition, suppose $\mathfrak{b}$ is an isotropic subalgebra of $\FG$.
Consider the affine vertex Lie algebra $(v_k(\FG),Y^{-}(?,z),T)$, as a vector space $v_k(\FG)$ that can be factorized as $v_k(\FG)=v_k(\mathfrak{a})\oplus v_k'(\mathfrak{b})$ where $v'_{k}(\mathfrak{b}):=\Span_{\C}\{xt^n.v_k\mid x\in\mathfrak{b}, n\leq-1\}$. Moreover, $(v_k'(\mathfrak{b}),Y^{-}(?,z),T)$ forms a vertex Lie subalgebra of $v_k(\FG)$. This motivates the following theorem.
\begin{thm}\label{thm:RM}
Let $\FG$ be a finite-dimensional Lie algebra as above with Lie algebra decomposition $\FG=\mathfrak{a}\oplus\mathfrak{b}$ where $\mathfrak{b}$ is an isotropic subalgebra of $\FG$. Define $R_{L}:=\frac{1}{2}(P_{+}-P_{-})$ where $P_{+},\ P_{-}$ are the projection operators from $v_{k}(\mathfrak{g})$ to $v_{k}(\mathfrak{a})$ and $v_{k}'(\mathfrak{b})$, respectively. Then $R_L$ is a classical $R$-matrix of vertex Lie algebra $v_k(\FG)$. Following Lie bialgebra's terminology, we shall call the $R$-matrix obtained in this way the factorizable $R$-matrix of vertex Lie algebra.
\end{thm}
\begin{rmk}\label{rmk:level0rmatrix}
Let $\FG$ be a finite-dimensional Lie algebra with a Lie algebra decomposition $\FG=\mathfrak{a}\oplus\mathfrak{b}$. Consider the affine vertex Lie algebra $v_0(\FG)$, it has a vertex Lie algebra decomposition $v_0(\FG)\cong v_0(\mathfrak{a})\oplus v_0'(\mathfrak{b})$. Therefore, for any Lie algebra decomposition $\FG\cong\mathfrak{a}\oplus\mathfrak{b}$, we have a factorizable R-matrix for the level $0$ vertex Lie algebra of $\FG$. In this case, we simply say that $R_L$ is an $R$-matrix associated with Lie algebra decomposition $\mathfrak{a}\oplus\mathfrak{b}$.
\end{rmk}
Recall a finite-dimensional Manin triple which is a triple of finite-dimensional Lie algebras  $(\FG,\FG_+,\FG_-)$, where $\FG$ has an invariant bilinear form $(~\mid~)$ such that
\begin{itemize}
\item[i.] $\FG=\FG_+\oplus\FG_-$ is a Lie algebra decomposition.
\item[ii.] $\FG_+$ and $\FG_-$ are isotropic Lie subalgebras of $\FG$ with respect to $(~\mid~)$.
\end{itemize}
This gives the following example immediately.
\begin{cor}\label{cor:FRMT}
Let $(\FG,\FG_+,\FG_-)$ be a finite-dimensional Manin triple with an invariant bilinear form $(~\mid~)$. Then the affine vertex Lie algebra $v_k(\FG)$ admits a vertex Lie algebra factorization $v_k(\FG)\cong v_k(\FG_+)\oplus v_k'(\FG_-)$. In addition, $v_k(\FG)$ has a factorizable $R$-matrix $R_L:=\frac{1}{2}(P_+-P_-)$, where $P_+$, $P_-$ are the projection operators from $v_k(\FG)$ to $v_k(\FG_+)$ and $v_k'(\FG_-)$, respectively.
\end{cor}
\begin{dfn}\label{dfn:RM}
Let $R_{L}$ be an $R$-matrix of vertex Lie algebra $v_{k}(\mathfrak{g})$ obtained by Theorem \ref{thm:RM}. Then, we say the $R_{L}$ is an $R$-matrix of $v_{k}(\mathfrak{g})$ associated with the Lie algebra decomposition $\mathfrak{a}\oplus\mathfrak{b}$.
\end{dfn}
We will use the following notations in this subsection. Let $\mathfrak{g}$ be a semi-simple Lie algebra with a bilinear, non-degenerate, $ad$-invaraint form $(\cdot,\cdot)$ and $\mathfrak{h}$ be its Cartan subalgebra. Then, its root space decomposition reads as $$\mathfrak{g}=\mathfrak{h}\oplus\bigoplus_{\alpha\in\Delta}\mathfrak{g}_{\alpha},\hspace{1cm} \mathfrak{g}_{\alpha}:=\{y \in \mathfrak{g}:[x, y]=\alpha(x) y \text { for all } x \in \mathfrak{h}\},$$ where $\Delta$ is the set of roots of $\mathfrak{g}$. Let $\Pi:=\{\alpha_{1},\alpha_{2},\ldots,\alpha_{r}\}$ be the set of simple roots of $\mathfrak{g}$ where $r=\dim\mathfrak{h}$. Also, we let $\Delta_{+},\ \Delta_{-}$ be the set of positive roots, and negative roots of $\mathfrak{g}$, respectively. Then the triangle decomposition of $\mathfrak{g}$ can be written as $$\mathfrak{g}=\mathfrak{n}_{-}\oplus\mathfrak{h}\oplus\mathfrak{n}_{+}=\mathfrak{n}_{-}\oplus\mathfrak{b},$$ where $\mathfrak{n}_{\pm}=\oplus_{\alpha\in\Delta_{\pm}}\mathfrak{g}_{\alpha}$ and $\mathfrak{b}:=\mathfrak{h}\oplus\mathfrak{n}_{+}$ is called the Borel subalgera of $\mathfrak{g}$. Since $\dim\mathfrak{g}_{\alpha}=1$, one chooses a non-zero element $e^{\alpha}\in\mathfrak{g}_{\alpha},\ \forall\alpha\in\Delta$ and $h^{i}\in\mathfrak{h},\ i=1,2,\cdots,r$ as generators of Lie algebra $\mathfrak{g}$.
\begin{exa}\label{exa:BD}
Consider the Borel decomposition of a rank $r$ finite-dimensional Lie algebra $\mathfrak{g}$, $\mathfrak{g}=\mathfrak{b}\oplus\mathfrak{n}_{-}$ with $\dim\mathfrak{n}_{\pm}=n$ and the $R$-matrix $R_{L}$ associated with it. The corresponding vertex Lie algebra factorization is $v_{k}(\mathfrak{g})=v_{k}(\mathfrak{b})\oplus v_{k}'(\mathfrak{n}_{-})$ where $$v_{k}(\mathfrak{b})=\mathrm{Span}_{\mathbb{C}}\{v_{k},h^{i}_{-p}.v_{k},e^{\alpha_{j}}_{-p}.v_{k}\mid i=1,2,\ldots, r,\ j=1,2,\ldots,n,p\geq0\},$$ and $$v_{k}'(\mathfrak{n}_{-})=\mathrm{Span}_{\mathbb{C}}\{e^{-\alpha_{i}}_{-p}.v_{k}\mid i=1,2,\ldots,n,p\geq0\}.$$
Let $L:=v_{k}(\mathfrak{g})_{R_{L}}$, we compute the commutators of local Lie algebra $Lie(L)$ here. Note that we have identities,
\begin{align*}
&R_{L}(e^{-\alpha_{i}}_{-1}.v_{k})=\frac{1}{2}(P_{+}-P_{-})(f^{i}_{-1}.v_{k})=-\frac{1}{2}f^{i}_{-1}.v_{k},\\
&R_{L}(h^{j}_{-1}.v_{k})=\frac{1}{2}h^{j}_{-1}.v_{k},\\
&R_{L}(e^{\alpha_{i}}_{-1}.v_{k})=\frac{1}{2}e^{i}_{-1}.v_{k},
\end{align*}
$\forall i=1,2,\ldots n,j=1,2,\ldots r$. Hence, the operators $Y^{-}_{R_{L}}(?,z)$ read: $\forall i=1,2,\ldots,n$ and $j=1,2\ldots,r$,
\begin{align*}
Y^{-}_{R_{L}}(e^{-\alpha_{i}}_{-1}.v_{k},z)=\sum_{n\geq0}(-\frac{1}{2}e^{-\alpha_{i}}_{(n)}+e^{-\alpha_{i}}_{(n)}R_{L})z^{-n-1},
\end{align*}
\begin{align*}
Y^{-}_{R_{L}}(h^{j}_{-1}.v_{k},z)=\sum_{n\geq0}(\frac{1}{2}h^{j}_{(n)}+h^{j}_{(n)}R_{L})z^{-n-1},
\end{align*}
\begin{align*}
Y^{-}_{R_{L}}(e^{\alpha_{i}}_{-1}.v_{k},z)=\sum_{n\geq0}(\frac{1}{2}e^{\alpha_{i}}_{(n)}+e^{\alpha_{i}}_{(n)}R_{L})z^{-n-1}.
\end{align*}
We compute the Lie brackets of the local Lie algebra of $v_k(\FG)_{R_{L}}$. By formula \eqref{eq:LLA}, we have $\forall i,j=1,2,\ldots,n$, $p,q=1,2,\ldots,r$, and $m,k\in\Z$,
\begin{align*}
&[e^{-\alpha_{i}}_{[m]},e^{-\alpha_{j}}_{[k]}]_{R_{L}}=[e^{-\alpha_{i}}_{[m]},e^{-\alpha_{j}}_{[k]}],\\
&[e^{-\alpha_{i}}_{[m]},h^{j}_{[k]}]_{R_{L}}=0,\\
&[e^{-\alpha_{i}}_{[m]},e^{\alpha_{j}}_{[k]}]_{R_{L}}=0,\\
&[h^{p}_{[m]},h^{q}_{[k]}]_{R_{L}}=[h^{p}_{[m]},h^{q}_{[k]}],\\
&[h^{p}_{[m]},e^{\alpha_{j}}_{[k]}]_{R_{L}}=[h^{p}_{[m]},e^{\alpha_{j}}_{[k]}],\\
&[e^{\alpha_{i}}_{[m]},e^{\alpha_{j}}_{[k]}]_{R_{L}}=[e^{\alpha_{i}},e^{\alpha_{j}}]_{[m+k]}+m(e^{\alpha_{i}},e^{\alpha_{j}})K\delta_{m+k,0}.
\end{align*}
\end{exa}
The calculations in Example \ref{exa:BD} give the following statement.
\begin{thm}
For a semi-simple Lie algebra $\mathfrak{g}$ with triangle decomposition $\mathfrak{g}=\mathfrak{n}_{-}\oplus\mathfrak{h}\oplus\mathfrak{n}_{+}$, let $R_{L}$ be the $R$-matrix associated with the Borel decomposition $\mathfrak{g}=\mathfrak{n}_{-}\oplus\mathfrak{b}$. Denote the universal enveloping vertex algebra of the $R_{L}$-twisted vertex Lie algebra $v_k(\FG)_{R_{L}}$ by $V^{k}(\FG)_{R_{L}}$.
Then, generators of fields in $V^{k}(\mathfrak{g})_{R_{L}}$ have the following OPEs,
\begin{align*}
&e^{-\alpha_{i}}_{R_{L}}(z)e^{-\alpha_{j}}_{R_{L}}(w)\sim-\frac{[e^{-\alpha_i},e^{-\alpha_j}](w)}{z-w}-\frac{(e^{\alpha_i},e^{\alpha_j})K}{(z-w)^{2}},\\
&e^{-\alpha_{i}}_{R_{L}}(z)h^{p}_{R_{L}}(w)\sim0,\\
&e^{-\alpha_{i}}_{R_{L}}(z)e^{\alpha_{j}}_{R_{L}}(w)\sim0,\\
&h^{p}_{R_{L}}(z)h^{q}_{R_{L}}(w)\sim\frac{(h^{p},h^{q})K}{(z-w)^{2}},\\
&h^{p}_{R_{L}}(z)e^{\alpha_{j}}_{R_{L}}(w)\sim\frac{\alpha_{j}(h^{p})e^{\alpha_{j}}(w)}{z-w}+\frac{(h^{p},e^{\alpha_{j}})K}{(z-w)^{2}},\\
&e^{\alpha_{i}}_{R_{L}}(z)e^{\alpha_{j}}_{R_{L}}(w)\sim\frac{[e^{\alpha_{i}},e^{\alpha_{j}}](w)}{z-w}+\frac{(e^{\alpha_{i}},e^{\alpha_{j}})K}{(z-w)^{2}}.
\end{align*}
If we use the $\lambda$-bracket,
\begin{align*}
&[e^{-\alpha_i}{_\lambda}e^{-\alpha_j}]_{R_{L}}=[e^{-\alpha_{i}},e^{-\alpha_j}]-(e^{\alpha_i},e^{\alpha_j})\lambda,\\
&[e^{-\alpha_i}{_\lambda}h^{p}]_{R_{L}}=0,\\
&[e^{-\alpha_i}{_\lambda}e^{\alpha_j}]_{R_{L}}=0,\\
&[h^{p}{_\lambda}h^{q}]_{R_{L}}=(h^{p},h^{q})K\lambda,\\
&[h^{p}{_\lambda}e^{\alpha_j}]_{R_{L}}=\alpha_j(h^{p})e^{\alpha_j}+(h^{p},e^{\alpha_{j}})K\lambda,\\
&[e^{\alpha_i}{_\lambda}e^{\alpha_j}]_{R_{L}}=[e^{\alpha_i},e^{\alpha_j}]+(e^{\alpha_i},e^{\alpha_j})K\lambda,
\end{align*}
$\forall i,j=1,2,\ldots,n$ and $p,q=1,2,\ldots,r$.
\end{thm}
\begin{exa}\label{exa:Iwa}
Consider a simple Lie algebra $\mathfrak{g}$ with the Iwasawa decomposition $\mathfrak{g}=\mathfrak{a}\oplus\mathfrak{b}$ where 
$$\mathfrak{a}=\mathrm{Span}_{\mathbb{C}}\{e^{\alpha}+e^{-\alpha}\mid\alpha\in\Delta_{+}\}\oplus\mathfrak{h},$$
$$\mathfrak{b}=\mathrm{Span}_{\mathbb{C}}\{e^{\alpha}-e^{-\alpha}\mid\alpha\in\Delta_{+}\}.$$
At level $0$, we have a vertex Lie algebra factorization $v_0(\FG)=v_0(\mathfrak{a})\oplus v_0'(\mathfrak{b})$. Let $R_{L}$ be the $R$-matrix associated with the Iwasawa decomposition and $L=v_{0}(\mathfrak{g})$. Then the operators act on the generators as
\begin{align*}
&R_{L}(e_{-1}^{-\alpha_{i}}.v_{k})=\frac{1}{2}e^{\alpha_{i}}_{-1}.v_{k},\\
&R_{L}(h^{j}_{-1}.v_{k})=\frac{1}{2}h^{i}_{-1}.v_{k},\\
&R_{L}(e_{-1}^{\alpha_i}.v_{k})=\frac{1}{2}e_{-1}^{-\alpha_i}.v_{k},
\end{align*}
$\forall\ i=1,2,\ldots,n,$ and $j=1,2,\ldots r$, and the operators $Y^{-}_{R_{L}}(?,z)$ are given by
\begin{align*}
&Y^{-}_{R_{L}}(e^{-\alpha_{i}}_{-1}.v_{k},z)=\sum_{n\geq0}(\frac{1}{2}e^{\alpha_{i}}_{(n)}+e^{-\alpha_{i}}_{(n)}R_{L})z^{-n-1},\\
&Y^{-}_{R_{L}}(h^{j}_{-1}.v_{k},z)=\sum_{n\geq0}(\frac{1}{2}h^{j}_{(n)}+h^{j}_{(n)}R_{L})z^{-n-1},\\
&Y^{-}_{R_{L}}(e^{\alpha_{i}}_{-1}.v_{k},z)=\sum_{n\geq0}(\frac{1}{2}e^{-\alpha_{i}}_{(n)}+e^{\alpha_{i}}_{(n)}R_{L})z^{-n-1}.
\end{align*}
Hence, by formula \eqref{eq:LLA} the local Lie algebra $\mathrm{Lie}(L_{R_{L}})$ satisfies the following commutation relations: $\forall i,j=1,2,\ldots,n,\ p,q=1,2,\ldots,r,m,n\in\Z$,
\begin{align}
&[e^{-\alpha_{i}}_{[m]},e^{-\alpha_{j}}_{[k]}]_{R_{L}}=\frac{1}{2}([e^{\alpha_{i}},e^{\alpha_{j}}]_{[m+k]}+[e^{-\alpha_{i}},e^{\alpha_{j}}]_{[m+k]})\nonumber\\
&[e^{-\alpha_{i}}_{[m]},h^{p}_{[k]}]_{R_{L}}=-\frac{1}{2}\alpha_{i}(h^{p})(e^{\alpha_{i}}_{[m+k]}-e^{-\alpha_{i}}_{[m+k]})\nonumber\\
&[e^{-\alpha_{i}}_{[m]},e^{\alpha_{j}}_{[k]}]_{R_{L}}=\frac{1}{2}([e^{\alpha_{i}},e^{\alpha_{j}}]_{[m+k]}+[e^{-\alpha_{i}},e^{-\alpha_{j}}]_{[m+k]})\nonumber\\
&[h^{p}_{[m]},h^{q}_{[k]}]_{R_{L}}=0\nonumber\\
&[h^{p}_{[m]},e^{\alpha_{j}}_{[k]}]_{R_{L}}=\frac{1}{2}(\alpha_{j}(h^{p})e^{\alpha_{j}}_{[m+k]}-\alpha_{j}(h^{p})e^{-\alpha_{j}}_{[m+k]})\nonumber\\
&[e^{\alpha_{i}}_{[m]},e^{\alpha_{j}}_{[k]}]_{R_{L}}=\frac{1}{2}([e^{\alpha_{i}},e^{\alpha_{j}}]_{[m+k]})\nonumber
\end{align}
\end{exa}
The calculations in Example \ref{exa:BD} give the following statement.
\begin{thm}
For a semi-simple Lie algebra $\mathfrak{g}$ with Iwasawa decomposition $\mathfrak{g}=\mathfrak{n}_{-}\oplus\mathfrak{h}\oplus\mathfrak{n}_{+}$, let $R_{L}$ be the $R$-matrix associated with the Iwasawa decomposition $\mathfrak{g}=\mathfrak{n}_{-}\oplus\mathfrak{b}$. Consider the level $0$  universal enveloping vertex algebra of the $R_{L}$-twisted vertex Lie algebra $v_0(\FG)_{R_{L}}$ by $V^{0}(\FG)_{R_{L}}$.
Then, generators of fields in $V^{0}(\mathfrak{g})_{R_{L}}$ have the following OPEs,
\begin{align*}
&e_{R_L}^{-\alpha_i}(z)e_{R_L}^{-\alpha_j}(w)\sim\frac{1}{2}\frac{[e^{\alpha_i},e^{\alpha_j}](w)+[e^{-\alpha_i},e^{\alpha_j}](w)}{z-w}\\
&e^{-\alpha_i}_{R_L}(z)h_{R_{L}}(w)\sim-\frac{\alpha_i(h^p)}{2}\frac{e^{\alpha_i}(w)-e^{-\alpha_j}(w)}{z-w}\\
&e^{-\alpha_i}_{R_L}(z)e^{\alpha_j}_{R_L}(w)\sim\frac{1}{2}\frac{[e^{\alpha_i},e^{\alpha_j}](w)+[e^{-\alpha_i},e^{-\alpha_j}](w)}{z-w}\\
&h_{R_L}(z)h_{R_L}(w)\sim0\\
&h_{R_L}(z)e^{\alpha_j}_{R_L}(w)\sim\frac{1}{2}\frac{\alpha_{j}(h^{p})e^{\alpha_{j}}(w)-\alpha_{j}(h^{p})e^{-\alpha_{j}}(w)}{z-w}\\
&e^{\alpha_i}_{R_L}(z)e^{\alpha_j}_{R_L}(w)\sim\frac{1}{2}\frac{[e^{\alpha_i},e^{\alpha_j}](w)}{z-w}
\end{align*}
If we use the $\lambda$-bracket,
\begin{align*}
&[e^{-\alpha_i}{_\lambda}e^{-\alpha_j}]_{R_{L}}=\frac{1}{2}([e^{\alpha_i},e^{\alpha_j}]+[e^{-\alpha_i},e^{\alpha_j}]),\\
&[e^{-\alpha_i}{_\lambda}h^{p}]_{R_{L}}=-\frac{\alpha_i(h^p)}{2}(e^{\alpha_i}-e^{-\alpha_j}),\\
&[e^{-\alpha_i}{_\lambda}e^{\alpha_j}]_{R_{L}}=\frac{1}{2}([e^{\alpha_i},e^{\alpha_j}]+[e^{-\alpha_i},e^{-\alpha_j}]),\\
&[h^{p}{_\lambda}h^{q}]_{R_{L}}=0,\\
&[h^{p}{_\lambda}e^{\alpha_j}]_{R_{L}}=\frac{\alpha_j(h^p)}{2}(e^{\alpha_j}-e^{-\alpha_j}),\\
&[e^{\alpha_i}{_\lambda}e^{\alpha_j}]_{R_{L}}=\frac{1}{2}[e^{\alpha_i},e^{\alpha_j}],
\end{align*}
$\forall i,j=1,2,\ldots,n$ and $p,q=1,2,\ldots,r$.
\end{thm}
\section{Generalized AKS Scheme}\label{sec:AKS}
In this section, we prove the generalized AKS scheme, we first fix some notations here. In this section and the next section, we will consider the Poisson vertex algebra constructed in the following way, let $L$ be a vertex Lie algebra with $\lambda$-bracket $\{\cdot_{\lambda}\cdot\}$, and the Poisson vertex algebra $\V:=\Sym(L)$ with $\lambda$-bracket $\{\cdot_{\lambda}\cdot\}$ extended by the Master formula \eqref{MF}. Suppose $R_{L}$ is an $R$-matrix of $L,$ and the $\lambda$-bracket $\{\cdot_{\lambda}\cdot\}_{R_{L}}$ is twisted by $R_{L}$. Extend $\{\cdot_{\lambda}\cdot\}_{R_{L}}$ on $\V$ by the Master formula and denote this Poisson vertex algebra by $\V_{R_{L}}$. In this case, we will say the Poisson vertex algebra $\V$ has $\lambda$-bracket structures $\{\cdot_{\lambda}\cdot\}$ and $\{\cdot_{\lambda}\cdot\}_{R_{L}}$.
\begin{lemm}\label{lem7.1}
Consider a vertex Lie algebra $L$ and let $R_{L}$ be an R-matrix of it. Then, $\forall a,b\in L$, identity
\begin{equation*}
\{a_{\lambda}b\}_{R_{L}}=\{R_{L}(a)_{\lambda}b\}+\{a_{\lambda}R_{L}(b)\},	
\end{equation*}
holds.
\end{lemm}
\begin{proof}
$\forall a,b\in L,$ by the definition of $R$-matrix of vertex Lie algebra, we have 
\begin{align*}
Y^{-}_{R_{L}}(a,z)&=Y^{-}(R_{L}(a),x)+Y^{-}(a,z)R_{L}\\
	&=\sum_{n \geq 0}\left(R_L\left(a\right)_{(n)}\right) z^{-n-1}+\sum_{n \geq 0} a_{(n)} R_L z^{-n-1}\\
	&=\sum_{n \geq 0}\left[\left(R_L a\right)_{(n)}+a_{(n)} R_L\right] z^{-n-1}.
\end{align*}
Hence,
\begin{equation*}
	\left\{a_\lambda b\right\}_{R_L}=\sum_{n \geq 0}\left(R_L a\right)_{(n)}\left(b\right) \lambda^n+\sum_{n \geq 0} a_{(n)} R_L\left(b\right) \lambda^n.
\end{equation*}
On the other hand, we have
\begin{equation*}
	\left\{R_L (a) _{\lambda} b\right\}=\sum_{n \geq 0}\left(R_L a\right)_{(n)}\left(b\right) \lambda^n,
\end{equation*}
and
\begin{equation*}
\left\{a _\lambda R_L\left(b\right)\right\}=\sum_{n \geq 0} a_{(n)}\left(R_L\left(b\right)\right) \lambda^n	.
\end{equation*}
This completes the proof.
\end{proof}
\begin{notation}
Through Lemma \ref{lem7.2}--Theorem \ref{thm:GAKSS}, we will use the Einstein convention for the summation of the indices in $I$.
\end{notation}
\begin{lemm}\label{lem7.2}
	$\forall k,l\in I,\ r\in\mathbb{Z}_{+},$
\begin{align*}
	&\{u^{k,r}_{\lambda+\partial}u^{l}\}_{\rightarrow}=\{u^{k}_{\lambda+\partial}u^{l}\}_{\rightarrow}(-\lambda-\partial)^{r},\\
	&\{u^{l}_{\lambda+\partial}u^{k,r}\}_{\rightarrow}=(\partial+\lambda)^{r}\{u^{l}_{\lambda+\partial}u^{k}\}_{\rightarrow}	.
\end{align*}
\end{lemm}
\begin{proof}
By the Master formula, first, we have
\begin{align*}
	\left\{u_\lambda^{k, r} u^{l}\right\}&=\sum_{p,q\in\mathbb{Z}_{+}} \frac{\partial u^{l, 0}}{\partial u^{j,q}}(\partial+\lambda)^q\left\{u_{\lambda+\partial}^i u^j\right\}_{\rightarrow}(-\partial-\lambda)^p \frac{\partial u^{k, r}}{\partial u^{i, p}}\\
	&=\sum_{p, q\in\mathbb{Z}_{+}} \delta_j^l \delta_q^0(\partial+\lambda)^{q}\left\{u_{\lambda+\partial}^i u^j\right\}_{\rightarrow}(-\lambda-\partial)^p \delta_i^k \delta_p^r\\
	&=\left\{u^k_{\lambda}u^l\right\}(-\lambda)^r.
\end{align*}
This shows,
\begin{align*}
\{u^{k,r}_{\lambda+\partial}u^{l}\}_{\rightarrow}=\{u^{k,r}_{\lambda+\partial}u^{l}\}(-\lambda-\partial)^{r}.
\end{align*}
Hence, we proved the first equation. For the second equation, we compute $\{u^{l}_{\lambda}u^{k,r}\}$ first. By the Master formula and similar computation as above, we have
\begin{align*}
\left\{u_\lambda^l u^{k, r}\right\}&=\sum_{p, q\in\mathbb{Z}_{+}} \frac{\partial u^{k, r}}{\partial u^{j, q}}(\partial+\lambda)^{q}\left\{u_{\lambda+\partial}^i u^j\right\}_{\rightarrow}(-\partial-\lambda)^p \frac{\partial u^{l, 0}}{\partial u^{i, p}}=(\partial+\lambda)^r\left\{u_\lambda^l u^k\right\}.
\end{align*}
Therefore, $\{u^{l}_{\lambda+\partial}u^{k,r}\}_{\rightarrow}=(\partial+\lambda)^{r}\{u^{l}_{\lambda+\partial}u^{k}\}_{\rightarrow}$.
\end{proof}
\begin{IVL}\label{thm7.1}
Let $\mathcal{V}$ be a Poisson vertex algebra with $\lambda$-bracket$ \{\cdot_{\lambda}\cdot\}$ and $\{\cdot_{\lambda}\cdot\}_{R_{L}}$. Then $\forall f,g\in\mathcal{V},$ identity 
\begin{equation}\label{eq:twb}
	\{f_{\lambda}g\}_{R_{L}}=\sum_{p\in\mathbb{Z}_{+}}\{R_{L}(u^{i})_{\lambda}g\}(-\lambda-\partial)^{p}\frac{\partial f}{\partial u^{i,p}}+\sum_{q\in\mathbb{Z}_{+}}\frac{\partial g}{\partial u^{j,q}}(\partial+\lambda)^{q}\{f_{\lambda}R_{L}(u^{j})\},
\end{equation}
holds. Let $\mathfrak{z}(\V)$ be the $\lambda$-bracket center of $\V$ with respect to the $\lambda$-bracket $\{\cdot_{\lambda}\cdot\},$ then $\forall f,g\in\mathfrak{z}(\V),\ \{f_{\lambda}g\}_{R_{L}}=0$.
\end{IVL}
\begin{proof}[Proof of involution theorem]
By Lemma \ref{lem7.1} we have $\forall f,g\in\V,$
\begin{align*}
\{f_{\lambda}g\}_{R_{L}}=\sum_{p,q \in\mathbb{Z}_{+}}& \frac{\partial g}{\partial u^{j, q}}(\partial+\lambda)^q\left\{R_{L}\left(u^i\right)_{\lambda+\partial} u^j\right\}_{\rightarrow}(-\partial-\lambda)^p \frac{\partial f}{\partial u^{i, p}}\\
&+\sum_{p, q \in\mathbb{Z}_{+}} \frac{\partial g}{\partial u^{j, q}}(\partial+\lambda)^q\left\{u_{\lambda+\partial}^i R_{L}\left(u^j\right)\right\}_{\rightarrow}(-\partial-\lambda)^p \frac{\partial f}{\partial u^{i,p}}.
\end{align*}
Let 
\begin{equation*}
	(*)=\sum_{p,q \in\mathbb{Z}_{+}}\frac{\partial g}{\partial u^{j, q}}(\partial+\lambda)^q\left\{R_{L}\left(u^i\right)_{\lambda+\partial} u^j\right\}_{\rightarrow}(-\partial-\lambda)^p \frac{\partial f}{\partial u^{i, p}},
\end{equation*}
and
\begin{equation*}
	(**)=\sum_{p, q \in\mathbb{Z}_{+}} \frac{\partial g}{\partial u^{j, q}}(\partial+\lambda)^q\left\{u_{\lambda+\partial}^i R_{L}\left(u^j\right)\right\}_{\rightarrow}(-\partial-\lambda)^p \frac{\partial f}{\partial u^{i,p}}.
\end{equation*}
In order to show \eqref{eq:twb}, it suffices to verify that
\[(*)=\sum_{p\in\mathbb{Z}_{+}}\{R_{L}(u^{i})_{\lambda}g\}(-\lambda-\partial)^{p}\frac{\partial f}{\partial u^{i,p}},\]
and
\[(**)=\sum_{q\in\mathbb{Z}_{+}}\frac{\partial g}{\partial u^{j,q}}(\partial+\lambda)^{q}\{f_{\lambda}R_{L}(u^{j})\}.\]
Since $R_{L}(u^{i})$ lies in the vertex Lie algebra $L$, $R_{L}(u^{i})=\sum_{r\in\mathbb{Z}_{+}}c_{k,r}^{i}u^{k,r}.$ Thus, by Lemma \ref{lem7.1},
\begin{align*}
(*)&=\sum_{p,q \in\mathbb{Z}_{+}}\frac{\partial g}{\partial u^{j, q}}(\partial+\lambda)^q\left\{R_{L}\left(u^i\right)_{\lambda+\partial} u^j\right\}_{\rightarrow}(-\partial-\lambda)^p \frac{\partial f}{\partial u^{i, p}}\\
&=\sum_{p, q\in\mathbb{Z}_{+}} \sum_{r \in\mathbb{Z}_{+}} c^{i}_{k, r} \frac{\partial g}{\partial u^{j,q}}(\partial+\lambda)^q\left\{u_{\lambda+\partial}^{k,r} u^j\right\}_{\rightarrow}\frac{\partial f}{\partial u^{i, p}}	\\
&=\sum_{p, q\in\mathbb{Z}_{+}} \sum_{r \in\mathbb{Z}_{+}} c^{i}_{k, r} \frac{\partial g}{\partial u^{j,q}}(\partial+\lambda)^q\left\{u_{\lambda+\partial}^k u^j\right\}_{\rightarrow}(-\lambda-\partial)^r(-\partial-\lambda)^p \frac{\partial f}{\partial u^{i, p}}.	
\end{align*}
On the other hand, by the Master formula,
\begin{align*}
\sum_{p\in\mathbb{Z}_{+}}\{R_{L}(u^{i})_{\lambda}g\}&(-\lambda-\partial)^{p}\frac{\partial f}{\partial u^{i,p}}\\
&=\sum_{p,r\in\mathbb{Z}_{+}}\{c_{k,r}^{i}u^{k,r}_{\lambda}g\}(-\lambda-\partial)^{p}\frac{\partial f}{\partial u^{i,p}}	\\
&=\sum_{p,r\in\mathbb{Z}_{+}}\left(\sum_{s,q\in\mathbb{Z}_{+}}c_{k,r}^{i}\frac{\partial g}{\partial u^{j, q}}(\lambda+\partial)^{q}\left\{u_{\lambda+\partial}^t u^j\right\}_{\rightarrow}(-\lambda-\partial)^{s} \frac{\partial u^{k, r}}{\partial u^{t,s}}\right)(-\lambda-\partial)^{p}\frac{\partial f}{\partial u^{i,p}}\\
&=\sum_{p,r\in\mathbb{Z}_{+}}\left(\sum_{q\in\mathbb{Z}_{+}}c_{k,r}^{i}\frac{\partial g}{\partial u^{j, q}}(\lambda+\partial)^{q}\left\{u_{\lambda+\partial}^{k} u^j\right\}_{\rightarrow}(-\lambda-\partial)^{r}\right)(-\lambda-\partial)^{p}\frac{\partial f}{\partial u^{i,p}}\\
&=(*).
\end{align*}
Suppose $R_{L}(u^{j})=\sum_{r\in\mathbb{Z}_{+}}d_{k,r}^{j}u^{k,r},$ then by Lemma \ref{lem7.1},
\begin{align*}
(**)&=\sum_{q\in\mathbb{Z}_{+}} \frac{\partial g}{\partial u^{j, q}}(\partial+ \lambda)^q\left\{u^{j}_{\lambda} \sum_{r \in \mathbb{Z}_{+}} d_{k, r}^j u^{k, r}\right\}_{\rightarrow}(-\partial-\lambda)^p \frac{\partial f}{\partial u^{i, p}}\\
&=\sum_{p, q,r\in\mathbb{Z}_{+}} d_{k, r}^{j} \frac{\partial g}{\partial u^{j,q}}(\partial+\lambda)^q(\partial+\lambda)^{r}\left\{u_{\lambda+\partial}^j u^k\right\}_{ \rightarrow}(-\partial-\lambda)^p \frac{\partial f}{\partial u^{i, p}}.
\end{align*}
On the other hand, by the Master formula,
\begin{align*}
	\sum_{q\in\mathbb{Z}_{+}}\frac{\partial g}{\partial u^{j,q}}(\partial+\lambda)^{q}&\{f_{\lambda}R_{L}(u^{j})\}\\
 &=\sum_{q,r\in\mathbb{Z}_{+}}d_{k,r}^{j}\frac{\partial g}{\partial u^{j,q}}\{f_{\lambda}u^{k,r}\}\\
	&=\sum_{q,r,p,s\in\mathbb{Z}_{+}} \frac{\partial g}{\partial u^{j, q}}(\partial+\lambda)^q\left(d_{k, r}^j\frac{\partial u^{k, r}}{\partial u^{t, s}}(\partial+\lambda)^s\left\{u_{\lambda+\partial}^i u^t\right\} _{\rightarrow}(-\partial-\lambda)^p \frac{\partial f}{\partial u^{i, p}}\right)\\
&=\sum_{p,q,r\in\mathbb{Z}_{+}}d_{k, r}^{j}\frac{\partial g}{\partial u^{j,q}}(\partial+\lambda)^q(\partial+\lambda)^r\left\{u_{\lambda+\partial}^i u^k\right\}_{\rightarrow}(-\partial-\lambda)^p \frac{\partial f}{\partial u^{i, p}}\\
&=(**).
\end{align*}
Hence, we conclude that
\begin{align*}
\{f_{\lambda}g\}_{R_{L}}&=(*)+(**)\\
&=\sum_{p\in\mathbb{Z}_{+}}\{R_{L}(u^{i})_{\lambda}g\}(-\lambda-\partial)^{p}\frac{\partial f}{\partial u^{i,p}}+\sum_{q\in\mathbb{Z}_{+}}\frac{\partial g}{\partial u^{j,q}}(\partial+\lambda)^{q}\{f_{\lambda}R_{L}(u^{j})\}.
\end{align*}
In particular, as $f,g\in\grF(\mathfrak{z}(\widehat{\mathfrak{g}}))$, then $\{f_{\lambda}g\}_{R_{L}}=0$.
\end{proof}
Now, let $\FG$ be a finite-dimensional Lie algebra with a non-degenerate, $ad$-invariant bilinear form $\nondp$. Consider affine Kac-Moody vertex algebra $V^{k}(\FG)$ of $\FG$ at the critical level $k=-h^{\vee}$ where $h^{\vee}$ is the dual Coxter number of $\FG$. Then, by Theorem \ref{thm:Ar} and Remark \ref{rmk:level0} we have $\grF\mathfrak{z}(\widehat{g})\subset\grF V^{k}(\FG)\cong\C[\Jinf\FG^{*}]=\Sym'(v_0(\FG))$. Hence, in this case, we have the following scheme of integrability.
\begin{thm}[Generalized AKS theorem]\label{thm:GAKSS}
Let $\FG$ be a finite-dimensional Lie algebra with a non-degenerate, $ad$-invariant bilinear form $\nondp$. Consider affine Kac-Moody vertex algebra $V^{k}(\FG)$ of $\FG$ at the critical level $k=-h^{\vee}$, where $h^{\vee}$ is the dual Coxter number of $\FG$. Let $\V:=\grF V^{k}(\FG),\ R_{L}$ be a classical $R$-matrix of $v_0(\FG),$ then the Poisson vertex algebra $\V$ has two $\lambda$-bracket structures $\{\cdot_{\lambda}\cdot\}$ and $\{\cdot_{\lambda}\cdot\}_{R_{L}}$. Take $h_n\in\grF\mathfrak{z}(\kmg),\ n\in\Z_{+},$ then the sequence of functionals $\int h_n\in\V/\partial\V,\ n\in\Z$ satisfy $\{\int h_n,\int h_m\}_{R_{L}}=0$. This defines a Hamiltonian system $\frac{du^i}{dt_n}:=\{\int h_n,u^i\}$ where $i=1,2,\ldots,\dim\FG$ and $n\in\Z_{+}$.
\end{thm}
\begin{proof}
This scheme follows from the involution theorem, Theorem \ref{thm:Ar}, and Remark \ref{rmk:level0}.
\end{proof}
\subsection{Classical and generalized AKS scheme}\label{subsec:CGAKS} 
We show that in a special case, the generalized AKS scheme in the affine vertex Lie algebra case will be back to the classical AKS scheme of Lie algebra. Let $\FG$ be a finite-dimensional Lie algebra with a non-degenerate, $ad$-invariant bilinear form $\nondp$, and suppose $\FG$ has basis $u^i,\ i=1,2,\ldots\e$ where $\e=\dim\FG$. Consider the affine vertex Lie algebra associated with $\FG$ at level 0, i.e., $v_0(\FG)$. Denote the $\lambda$-bracket on the $v_{0}(\FG)$ by $\{\cdot_{\lambda}\cdot\}.$
Use the same notations as in Section \ref{subsec:CRLCRB}, the vertex Lie subalgebra $v_0(\FG)_0$ gives a Lie algebra with Lie bracket $\{\cdot_{\lambda}\cdot\}$. Recall as a vector space 
\[v_0(\FG)_0=\Span_{\C}\{u^i\mid i=1,2,\ldots,\e\}.\] In addition, as we indicated in Section \ref{subsec:CRLCRB}, the classical $R$-matrix of vertex Lie algebra $v_0(\FG)_0$ gives a classical $R$-matrix of Lie algebra $v_0(\FG)_0$. Let $L:=v_0(\FG)_0$ and $\V:=\Sym(L)$ be the Poisson vertex algebra with $\lambda$-bracket extended from $L$. Note that in this case $\forall f,g\in\V$ by the Master Formula,
\begin{equation}\label{eq:eqPoisson}
\{f_{\lambda}g\}=\frac{\partial g}{\partial u^{j}}\left\{u^{i}_{ \partial+\lambda} u^{j}\right\}_{\rightarrow}\frac{\partial f}{\partial u^{i}}=\frac{\partial g}{\partial u^j}\{u^i_{\lambda}u^j\}\frac{\partial f}{\partial u^i}.
\end{equation}
Since $v_0(\FG)_0\cong\FG$, consider $v_0(\FG)_0$ as a Lie algebra, the computation \eqref{eq:eqPoisson} shows that the $\lambda$-bracket structure on the $\Sym(L)$ is the same as Poisson bracket structure on the $\Sym(\FG)$. 
Therefore, we can recover all results in the classical AKS scheme by using the generalized AKS scheme.
\section{Applications}\label{sec:app}
In this section, we give several examples of our scheme. In particular, we find systems of PDEs, ODEs, and their Poisson structures associated with Poisson vertex algebras $\mathrm{Sym}'(v_{0}(\mathfrak{sl}_2))$ and $\mathrm{Sym}'(v_{0}(\mathfrak{sl}_3))$.
For convenience, we introduce the following definitions.
\begin{dfn}\label{dfn:RPS}
Let $v_k(\mathfrak{g})$ be an affine vertex Lie algebra, $R_{L}$ be an $R$-matrix of $v_{k}(\mathfrak{g})$, and $\mathrm{Sym}(v_{k}(\mathfrak{g}))$ be the Poisson vertex algebra equipped with the $\lambda$-bracket extended from $\lambda$-bracket of $v_{k}(\mathfrak{g})$ by the Master Formula. As we described in Section \ref{sec:AKS}, $\Sym(v_k(\FG))$ has two $\lambda$-bracket structures. Let $\mathcal{V}:=\Sym'(v_k(\FG))$, we obtain two $\lambda$-bracket inherited from $\Sym(v_k(\FG))$. We denote the Poisson structure on $\mathcal{V}$ given by the $\lambda$-bracket $\{\cdot_{\lambda}\cdot\}_{R_{L}}$ by $H_{R_{L}}(\partial)$. Also, we will denote the new Poisson vertex algebra by $\mathcal{V}_{R_L}$.
\end{dfn}
In the following, we will consider the affine vertex algebras at level $0$. Following the same notations as in Definition \ref{dfn:RPS}, consider an affine vertex Lie algebra $v_0(\FG)$, regarding Remark \ref{rmk:level0rmatrix}, any Lie algebra factorization $\FG\cong\mathfrak{a}\oplus\mathfrak{b}$ gives an $R$-matrix associated with Lie algebra decomposition $\mathfrak{a}\oplus\mathfrak{b}$. For simplicity, we will call the Poisson structure $H_{R_L}(\partial)$ on the $\mathcal{V}_{R_L}$ the Poisson structure associated with the Lie algebra decomposition $\mathfrak{a}\oplus\mathfrak{b}$. Recall that in the Examples \ref{exa:BD} and \ref{exa:Iwa}, we have denoted the basis elements of $v_k(\FG)$ by $e^{-\alpha_i}_{-1}.v_{k},h^{j}_{-1}.v_{k},e^{\alpha_i}_{-1}.v_{k}$.
In the following examples, we will denote the image of $e^{-\alpha_i}_{-1}.v_{k},h_j:=h^{j}_{-1}.v_{k},e_{i}:=e^{\alpha_i}_{-1}.v_{k}$ in the modified Poisson vertex algebra $\Sym'(v_0(\FG))$ by $f_i,h_j,e_i$, respectively.
\begin{exa}\label{exa:BSL2O}
By the data we computed in Example \ref{exa:BD}, the Poisson structure associated with the Borel decomposition of $\mathfrak{sl}_{2}$ is 
$$H_{R_{L}}(\partial)=\begin{bmatrix}0&0&0\\0&0&2e\\0&-2e&0\end{bmatrix}.$$
Take $\int h_n=\int S^{n}$ then the associated system is
\begin{align*}
&\frac{df}{dt_1}=0,\\
&\frac{dh}{dt_1}=4ef,\\
&\frac{de}{dt_1}=-2eh.
\end{align*}
This system admits a special solution:
$$f=1,\ h=\frac{1}{t_1+C},\ e=-\frac{1}{4}(\frac{1}{t_1+C})^{2},$$ where $C$ is an integration constant.
\end{exa}
\begin{exa}\label{exa:ISL2O}
By the data computed in Example \ref{exa:Iwa}, the Poisson structure associated with the Iwasawa decomposition of $\mathfrak{sl}_2$ is 
$$H_{R_{L}}(\partial)=\begin{bmatrix}0&f-e&0\\-f+e&0&e-f\\0&-e+f&0\end{bmatrix}.$$
Take $\int h_{n}:=\int S^{n},\ n=1,2,\cdots$ where $S$ is the Sugawara operator of $V^{-2}(\mathfrak{sl}_{2})$ i.e., $S=\frac{1}{2}h^{2}+2fe$. Then the first system of our hierarchy is
\begin{align*}
\frac{d f}{d t_{1}}=-h(f-e),\ \frac{d h}{d t_{1}}=2(f-e)(f+e),\ \frac{d e}{d t_{1}}=-h(f-e).
\end{align*}
This system is easy to solve, first note that $(f-e)_{t_{1}}=0$, hence we may choose $f-e=\frac{1}{2}$. Let $\xi:=f+e$, $\eta:=h$, then this system is equivalent to 
\begin{align*}
&\frac{d\xi}{d t_1}=-\eta,\\
&\frac{d\eta}{d t_1}=\xi.
\end{align*}
This is an oscillator equation (see e.g. \cite{LL1}): $$\frac{d^{2}\eta}{dt_1^{2}}+\frac{d\eta}{d t_1}=0.$$
\end{exa}
\begin{exa}\label{exa:BSL2P}
Take the same Poisson structure as in Example \ref{exa:BSL2O}.
Let $\int h_n:=\int SS^{(n)},\ n=1,2,\cdots$, in particular, note that $n=1$ gives a trivial flow. The second flow is,
\begin{equation}\label{eq:BSL2P}
\begin{aligned}
&\frac{d f}{dt_2}=0,\\
&\frac{d h}{dt_2}=-8 fe({h^{\prime}}^2+4 f^{\prime} e^{\prime}+2 e f^{\prime\prime}+hh^{\prime\prime}+2 f e^{\prime\prime}),\\
&\frac{d e}{d t_2}=4 he({h^{\prime}}^2+4 f^{\prime}e^{\prime}+2 e f^{\prime \prime}+hh^{\prime\prime}+2f e^{\prime\prime}).
\end{aligned}
\end{equation}
This system may be viewed as infinitely many integrable systems parameterized by $t_2$-independent function $f$. For example, let $f:=\frac{1}{8},$ then our system is equivalent to system,
\begin{align*}
&\frac{dh}{dt_2}=-e({h^{\prime}}^2+hh^{\prime\prime}+\frac{1}{4}e^{\prime\prime}),\\
&\frac{de}{dt_2}=4he({h^{\prime}}^{2}+hh^{\prime\prime}+\frac{1}{4}e^{\prime\prime}).
\end{align*}
Take $f=0$, then we may choose $h=\frac{1}{4}$, and our equations deduce to \begin{equation}\label{eq:BSL2P1}
\frac{de}{dt_2}=ee^{\prime\prime}.
\end{equation}
We seek a traveling wave solution of \eqref{eq:BSL2P1} with form $f(\xi)=f(x+ct_2)$. Pluging ansatz $f(\xi)$ into \eqref{eq:BSL2P1} separating variables and integrating, we obtain that, $$c\ln f-A=\frac{df}{d\xi},$$ where $A$ is an integration constant. Setting $A:=0$, and integrating it once more, we obtain that, 
$$\xi+\xi_0=\frac{1}{c}\Ei(\ln f),$$
where $\Ei(x)$ is the exponential integral, i.e., $\Ei(x):=-\int_{-x}^{\infty}\frac{e^{-t}}{t}dt$, see \cite{GT} and \cite{Ba1}. The function $\Ei(x)$ is a higher transcendental function \cite{Ba1}, see also \cite{WW}.
Also, we seek a traveling solution for the system \eqref{eq:BSL2P} in full generality. Note,
\[{h^{\prime}}^2+4 f^{\prime} e^{\prime}+2 e f^{\prime\prime}+hh^{\prime\prime}+2 f e^{\prime\prime}=\frac{1}{2}\partial^{2}_{x}h^2+2\partial^{2}_{x}fe,\]
and $\frac{d f}{d t_2}=0$, hence we may suppose $f(x,t_2)=\frac{1}{4}C(x).$
Let $\xi(x,t_2):=h(x,t_2),\ \eta(x,t_2):=e(x,t_2),$ then the system \eqref{eq:BSL2P} is equivalent to 
\begin{equation*}
\begin{aligned}
&\frac{d\xi}{dt_2}=-4C(x)\eta\partial_{x}^{2}(\frac{1}{2}h^2+\frac{1}{2}C(x)\eta),\\
&\frac{d\eta}{dt_2}=4\xi\eta\partial_{x}^{2}(\frac{1}{2}h^2+\frac{1}{2}C(x)\eta).
\end{aligned}
\end{equation*}
Therefore, $\frac{d}{dt_2}(\frac{1}{2}\xi^2+\frac{1}{2}C(x)\eta)=0.$ This means that quantity $\frac{1}{2}\xi^2+\frac{1}{2}C(x)\eta$ is a first integral of system \eqref{eq:BSL2P3}. Let $\frac{1}{2}B(x)^2:=\frac{1}{2}\xi^2+\frac{1}{2}C(x)\eta,$ then
\begin{align}
&\frac{d\xi}{dt_2}=-2C(x)(B(x)^{2})_{xx}\eta=-2(B(x)^2)_{xx}B(x)^2+2(B(x)^2)_{xx}\xi^2,\label{eq:BSL2P2}\\
&(\frac{d\eta}{dt_2})^2=4\xi^2\eta^2(B(x)^2)_{xx}^{2}=4\eta^2(B(x)^2)_{xx}^2(B(x)^2-C(x)\eta).\label{eq:BSL2P3}
\end{align}
From equation \eqref{eq:BSL2P3}, we find
\begin{equation}\label{eq:BSL2P4}
\frac{d\eta}{dt_2}=2\eta(B(x)^2-C(x)\eta)^{\frac{1}{2}}(B(x)^2)_{xx}.
\end{equation}
Separating variables and integrating \eqref{eq:BSL2P4} once, we have
\[\int\frac{1}{2\eta\sqrt{B^2-C(x)\eta}}d\eta=\int (B(x)^2)_{xx} dt_2.\]
This gives that
\begin{equation*}
-\frac{1}{B}\tanh^{-1}[\frac{\sqrt{B^2-C(x)\eta}}{B(x)}]=(B(x)^{2})_{xx}t_2+c(x),
\end{equation*}
where $c(x)$ is a function only depends on $x$.
Therefore, by rescaling the integration constant,
\begin{equation*}
\eta(x,t_2)=\frac{B^2(x)}{C(x)}\sech^2[-B(B^2)_{xx}t_2+c(x)]=\frac{B^2(x)}{C(x)}\sech^{2}\phi,
\end{equation*}
where $\phi=-B(B^2)_{xx}t_2+c(x).$ Using formula $\frac{1}{2}B(x)^2:=\frac{1}{2}\xi^2+\frac{1}{2}C(x)\eta,$ we find
\[\xi(x,t_2)=B(x)\tanh\phi.\]
Choose $c(x)=\frac{\sqrt{c}}{2}x$, where $c$ is a real number, $B(x)=\frac{9^{\frac{1}{3}}c^{\frac{1}{2}}}{2}(x-x_0)^{\frac{2}{3}}$,  where $x_0$ is a constant, and $C(x)=\frac{c}{2B(x)^2}$ then 
\begin{equation}\label{eq:preKdV}
\eta(x,t_2)=\frac{c}{2}\sech^{2}(\frac{\sqrt{c}}{2}(x-ct_2)).
\end{equation}
is coincident with the KdV one soliton. Also, in this case
\[\xi(x,t_2)=B(x)\tanh{(\frac{\sqrt{c}}{2}(x-ct_2))}\]
Rewrite equation \eqref{eq:BSL2P2} as 
\begin{equation}\label{eq:BSL2PN1}
\frac{d\xi}{d t_2}=E(x)\xi^2-F(x),
\end{equation}
where $E(x):=2(B(x)^2)_{xx},\ F(x):=2(B(x)^2)_{xx}B(x)^2$. We obtain a Ricatti equation whose coefficients depend on the parameter $x$. Hence, equation \eqref{eq:BSL2PN1} can be linearized: Let $\tau(t_2)$ be a function satisfying $\xi(t_2)=-\frac{1}{E(x)}\frac{d}{d t_2}\log\tau(t_2)$. Then, the equation \eqref{eq:BSL2PN1} is equivalent to
\begin{equation}\label{eq:BSL2PN2}
\frac{d^2\tau}{d t_2^2}-P(x)^2\tau=0,
\end{equation}
where $P(x)=2(B(x)^{2})_{xx}B(x)$. Take two linearly independent solution of equation \eqref{eq:BSL2PN2}, say $\tau_1(x,t_2)=\exp{(-c(x)+P(x)t_2)}$ and $\tau_2(x,t_2)=\exp{(c(x)-P(x)t_2)}$ where $c(x)$ is an arbitrary function only depends on variable $x$. We see all solutions of \eqref{eq:BSL2PN2} have form
\[\tau(x,t_2)=c_1e^{-c(x)+P(x)t_2}+c_2e^{c(x)-P(x)t_2},\]
for some $c_1,c_2\in\C$. This means all solutions of \eqref{eq:BSL2PN1} have form
\[\xi(x,t_2)=-\frac{1}{E(x)}\frac{c_1\tau'_1(t_2)+c_2\tau'_2(t_2)}{c_1\tau_1(t_2)+c_2\tau_2(t_2)},\]
where $\prime$'s are derivative with respect to the variable $t_2$.
Since this solution is completely determined by the ratio of $c_1,c_2$, we may identify them with points on $\mathbb{CP}^{1}$. Note that by the choice of $B(x)$, $E(x)=-\frac{1}{B(x)}$. Hence,
\begin{align*}
\xi(x,t_2)&=-\frac{B(x)}{P(x)}\frac{c_1P(x)\tau_1(t_2)-c_2P(x)\tau_2(t_2)}{c_1\tau_1(t_2)-c_2\tau_2(t_2)}\\
&=B(x)\frac{-c_1 e^{-\zeta}+c_2 e^{\zeta}}{c_1 e^{-\zeta}+c_2 e^{\zeta}},
\end{align*}
where $\zeta:=c(x)-P(x)t_2$. Note that as $[c_1:c_2]=[1:1],$ $\xi(x,t_2)=B(x)\tanh\zeta$. Therefore, by identity $\frac{1}{2}B(x)^2:=\frac{1}{2}\xi^2+\frac{1}{2}C(x)\eta,$
\[\eta(x,t_2)=\frac{B(x)^{2}}{C(x)}\sech^2\zeta.\]
Let $c\in\mathbb{R}$, choose $c(x):=\frac{\sqrt{c}}{2}x$, $B(x):=(-\frac{9}{16})^{\frac{1}{3}}c^{\frac{1}{2}}(x-x_0)^{\frac{2}{3}}$, where $x_0$ is a constant. We find
\begin{equation*}
\eta(x,t_2)=\frac{c}{2}\sech^2(\frac{\sqrt{c}}{2}(x-ct_2)).
\end{equation*}
We see this solution as the same as equation \eqref{eq:preKdV}.
\end{exa}
\begin{exa}\label{exa:ISL2P}
Take the same Poisson structure as in Example \ref{exa:ISL2O}.
Let $\int h_n:=\int SS^{(n)},$ $\ n=1,2,\ldots$, in particular, note that $n=1$ gives a trivial system. The second system, i.e., $n=2$, is
\begin{equation*}
\begin{aligned}
&\frac{df}{dt_2}=2 h(f-e)({h^{\prime}}^2+4 f^{\prime} e^{\prime}+2 e f^{\prime \prime}+hh^{\prime \prime}+2 f e^{\prime \prime}),\\
&\frac{dh}{dt_2}=-4(f-e)(f+e)({h^{\prime}}^2+4 f^{\prime} e^{\prime}+2 e f^{\prime \prime}+hh^{\prime \prime}+2 fe^{\prime \prime}),\\
&\frac{de}{dt_2}=2 h(f-e)({h^{\prime}}^2+4 f^{\prime} e^{\prime}+2 ef^{\prime \prime}+hh^{\prime \prime}+2f e^{\prime \prime}).
\end{aligned}
\end{equation*}
Observing this system, we see the first and third equations guarantee that $f-e$ is a first integral. Therefore, we may put $f-e=C(x).$ Let
\begin{equation}\label{eq:newvar}
\xi(x,t_2):=h(x,t_2),\ \eta(x,t_2):=f(x,t_2)+e(x,t_2),
\end{equation}
and note
\[{h^{\prime}}^2+4 f^{\prime} e^{\prime}+2 e f^{\prime \prime}+hh^{\prime \prime}+2 f e^{\prime \prime}=\partial_x^{2}(\frac{1}{2}h^2+2fe).\]
In addition, Equation \eqref{eq:newvar} implies $2fe=\frac{1}{2}\eta^2-\frac{1}{2}C(x)^2.$
Hence, our system is equivalent to
\begin{align*}
&\frac{d\xi}{dt_2}=-4C(x)\eta\partial_x^{2}(\frac{1}{2}\xi^2+\frac{1}{2}\eta^2-\frac{1}{2}C(x)^2),\\
&\frac{d\eta}{dt_2}=4C(x)\xi\partial_x^{2}(\frac{1}{2}\xi^2+\frac{1}{2}\eta^2-\frac{1}{2}C(x)^2).
\end{align*}
This implies that
\[\xi\frac{d\xi}{dt_2}=-\eta\frac{d\eta}{dt_2},\]
and hence $\frac{d}{dt_2}(\frac{1}{2}\xi^2+\frac{1}{2}\eta^2)=0.$ Namely, $\frac{1}{2}\xi^2+\frac{1}{2}\eta^2$ is a first integral. Let $\frac{1}{2}\xi^2+\frac{1}{2}\eta^2=\frac{1}{2}A(x)^2$ then
\begin{align*}
&\frac{d\xi}{dt_2}=-4C(x)P_{xx}\eta,\\
&\frac{d\eta}{dt_2}=4C(x)P_{xx}\xi,
\end{align*}
where $P(x):=\frac{1}{2}A(x)^2-\frac{1}{2}C(x)^2.$ Note that if $A(x)=C(x)$ we get trivial equation. Hence, we may suppose $A(x)\neq C(x).$
The identity $\frac{1}{2}\xi^2+\frac{1}{2}\eta^2=\frac{1}{2}A(x)^2$ gives
\begin{equation}\label{eq:ISL2P2}
\frac{d\xi}{dt_2}=-4C(x)P_{xx}\sqrt{A(x)^2-\xi^2}.
\end{equation}
Use separation of variables and integrate \eqref{eq:ISL2P2} once. We find
\begin{equation*}
\int\frac{d\xi}{\sqrt{A(x)^2-\xi^2}}=-4C(x)P_{xx}\int dt_2.
\end{equation*}
Therefore,
\[\mathrm{tan}^{-1}(\frac{\xi}{\sqrt{A(x)^2-\xi^2}})=-4C(x)P_{xx}(t_2+\zeta_0(x)),\]
where $\zeta_0(x)$ does not depend on $t_2$.
Let $\psi(x,t_2)=-4C(x)P_{xx}(t_2+\zeta_0)$ then
\begin{equation}\label{eq:TrSol}
\xi(x,t_2)=\frac{A(x)\mathrm{tan}\psi}{\sqrt{1+\mathrm{tan}^{2}\psi}}=\frac{A(x)\mathrm{tan}\psi}{\mathrm{sec}\psi}=A(x)\mathrm{sin}\psi.
\end{equation}
Using the identity $\frac{1}{2}\xi^2+\frac{1}{2}\eta^2=\frac{1}{2}A(x)^2,$ we find 
\begin{equation}\label{eq:TrSol1}
\eta(x,t_2)=A(x)\cos\psi.
\end{equation}
One writes solution \eqref{eq:TrSol} in terms of a traveling wave solution. To do this, suppose $A(x)$ and $-4C(x)P_{xx}=\alpha$ are constants with respect to both $x,t_2$. Since $A(x)$ is a constant with respects to the variable $x$,
\[-4C(x)P_{xx}=-4C(x)\frac{\partial^2}{\partial x^2}(\frac{1}{2}A^{2}-\frac{1}{2}C(x)^{2})=2C(x)\frac{\partial^{2}C(x)^{2}}{\partial x^2}=\alpha.\] 
Choose $C(x)=\frac{(9\alpha)^{\frac{1}{3}}}{2}(x-x_0)^{\frac{2}{3}}$ and $\zeta_0(x):=\frac{x}{\alpha},$ then from \eqref{eq:TrSol} and \eqref{eq:TrSol1} we have
\[\xi(x,t_2)=A\sin(x+\alpha t_2),\]
and
\[\eta(x,t_2)=A\cos(x+\alpha t_2).\]
\end{exa}
\section{Appendix I: The first few hierarchies of $\fraksl_3$}\label{sec:APPIII}
\subsection{Poisson structure associated with the Borel decomposition of $\fraksl_3$}
The Poisson structure associated with the Borel decomposition of $\mathfrak{sl}_3$ is
$$H_{R_{L}}(\partial)=\begin{bmatrix}
0 & f_3 & 0 & 0 & 0 & 0 & 0 & 0 \\
-f_3 & 0 & 0 & 0 & 0 & 0 & 0 & 0 \\
0 & 0 & 0 & 0 & 0 & 0 & 0 & 0 \\
0 & 0 & 0 & 0 & 0 & 2 e_1 & -e_2 & e_3 \\
0 & 0 & 0 & 0 & 0 & -e_1 & 2 e_2 & e_3 \\
0 & 0 & 0 & -2 e_1 & e_1 & 0 & e_3 & 0 \\
0 & 0 & 0 & e_2 & -2 e_2 & -e_3 & 0 & 0 \\
0 & 0 & 0 & -e_3 & -e_3 & 0 & 0 & 0
\end{bmatrix}.$$
Let $S_{1}$ be the quadratic Sugawara operator of $\mathfrak{sl}_2$ and $\int h_n=\int S_1^{n},$ then the first system is given by 
\begin{align*}
&\frac{d f_1}{d t_1}=-2f_3e_2,\ \frac{d f_2}{d t_1}=2f_3e_1,\\
&\frac{d f_3}{d t_1}=0,\ \frac{d h_1}{dt_1}=-4 f_1e_1+2f_2e_2-2f_3e_3,\\
&\frac{d h_2}{d t_1}=2 f_1e_1-2f_2e_2-2f_3e_3,\ \frac{d e_1}{d t_1}=2 h_1e_1-2f_2e_3,\\
&\frac{d e_2}{d t_1}=\frac{2}{3}(-h_1e_2+h_2e_2+3f_1e_3),\ \frac{d e_3}{d t_1}=2(h_1+h_2) e_3.
\end{align*}
Take the cubic Sugawara operator $S_2$ of $\mathfrak{sl}_3$ and $\int h_n=\int S_2^{n}$, the first system is a system of ODEs that takes the form
\begin{align*}
\frac{d f_1}{d t_1}&=\frac{1}{3} f_3\left(2 h_1e_2+h_2e_2-3 f_1e_3\right),\\
\frac{d f_2}{d t_1}&=f_3(h_1e_1+f_2e_3),\\
\frac{d f_3}{d t_1}&=0,\\
\frac{d h_1}{d t_1}&=\frac{1}{3}(-f_2(2 h_1+h_2)e_2+f_3(-h_1+h_2) e_3-3f_1(2h_1e_1+f_2e_3)),\\
\frac{d h_2}{d t_1}&=\frac{1}{3}(f_2(2h_1+h_2)e_2+f_3(-h_1+h_2) e_3+3 f_1(h_1e_1-h_1e_3)),\\
\frac{d e_1}{d t_1}&=\frac{1}{3}(h_1^2 e_1+6 f_1e_1^2-3 f_2e_1e_2+f_2 h_2e_3+ 3f_3e_1(-h_1+2h_2+e_3)+2 h_1(h_2e_1+f_2e_3)),\\
\frac{d e_2}{d t_1}&=\frac{1}{9}(-h_1^2 e_2+e_2(-h_2^2-9 f_1e_1+3 f_2e_2)+
3 f_3e_2(3h_1-3h_2-2e_3)+h_1(-4h_2e_2+9 f_1e_3)),\\
\frac{d e_3}{d t_1}&=\frac{1}{3}(h_1^2-h_2^2+3 f_3(h_1+h_2)+3 f_1e_1-3f_2e_2)e_3.
\end{align*}
\newpage
\subsection{Poisson structure associated with the Iwasawa decomposition of $\fraksl_3$}
By the computations in Example \ref{exa:Iwa}, the Poisson structure associated with the Iwasawa decomposition of $\mathfrak{sl}_3$ is:
\[
 \renewcommand{\arraystretch}{1.3}
 \frac{1}{2}\begin{bmatrix}
 0 & -h_1 & f_2-e_2 & -2f_1+2e_1 &f_1+e_1 & 0 &-e_3 & 0\\
h_1&0&f_1+e_1&f_2-e_2&f_2-e_2&-f_3+e_8&0&0\\
-f_2+e_2&-f_1-e_1&0&-f_3+e_3&-f_3+e_3&0&0&0\\
2f_1-2e_1&-f_2+e_2&f_3-e_3&0&0&2f_1-2e_1&f_2-e_2&f_3-e_3\\
-f_1-e_1&-f_2+e_2&f_2-e_3&0&0&-f_1-e_1&2f_2-2e_2&f_3-e_8\\
0&f_3-e_3&0&-2f_1+2e_1&f_1+e_1&0&0&f_2-e_2\\
e_3&0&0&-f_2-e_2&-f_2+e_2&0&0&-f_1+e_1\\
0&0&0&-f_3+e_3&-f_3+e_3&-f_2+e_2&f_1-e_1&0
\end{bmatrix}
\]
Let $S_1,\ S_2$ be the quadratic and cubic Sugawara operators of $V^{-3}(\mathfrak{sl}_3)$ respectively. Let $\int h^{i}_n:=\int S_i^{n}$. Then, the first system, i.e., $n=1$ of $\int h_n^{1}$ is
\begin{align*}
\frac{df_1}{dt_1}&=\frac{1}{3}(-3 f_1 h_1+4 h_2e_1+h_1(5 e_1-3 e_2)-3 e_2 e_3),\\
\frac{df_2}{dt_1}&=f_2(h_1+h_2)+h_1 e_1-h_1 e_2-h_2e_2-f_1(f_3-2 e_3)+e_1 e_3,\\
\frac{df_3}{dt_1}&=-f_3(h_1+h_2)-f_2e_1-f_1 e_2+h_1 e_3+h_2 e_3,\\
\frac{dh_1}{dt_1}&=2 f_1^2+f_2^2+f_3^2-2 e_1^2+e_2^2-e_3^2,\\
\frac{dh_2}{dt_1}&=-f_1^2+2 f_2^2+f_3^2-2 f_1 e_1-e_1^2-3 f_2 e_2+e_2^2-e_3^2,\\
\frac{de_1}{dt_1}&=f_2f_3-f_1h_1+\frac{5}{3} h_1 e_1+\frac{4}{3} h_2 e_1-e_2 e_3,\\
\frac{de_2}{dt_1}&=-f_1f_3-\frac{1}{3}f_2(4 h_1+5 h_2)+f_3 e_1+h_2 e_2+e_1e_3,\\
\frac{de_3}{dt_1}&=-f_3(h_1+h_2)-f_2e_1+f_1e_2+h_1e_3+h_2e_3.
\end{align*}
The first hierarchy of $\int h_n^{2}$ , i.e., $n=1$ is,
\begin{align*}
\frac{df_1}{dt_1}&=\frac{1}{6}(-6 {f_1}^2 e_1-3f_3h_1e_1+h_1^2 e_1+6 f_3 h_1e_1+2 h_1h_2e_1)\\
&-3h_1h_2e_2+f_1(-h_1^2-2h_1h_2+6 e_1^2+3f_2e_2+3f_3(h_1-2h_2-e_3))\\
&+3 f_3e_1e_3-h_1e_2e_3+h_2 e_2e_3-f_2(f_3(2 h_1+h_2)+3 e_1e_2-3h_1(h_2+e_3))),\\
\frac{df_2}{dt_1}&=\frac{1}{6}(3 h_1h_2e_1+6f_3h_1e_2-3 f_3 h_1e_2-2h_1h_2e_2-h_2e_2+h_1e_1e_3-h_2e_1e_3-3f_3e_2e_3\\&+f_2(h_2(2 h_1+h_2)+3 f_3(-2h_1+h_2+e_3))+f_1(-3 f_3h_1+3f_2e_1
-3 e_1 e_2+h_2e_3+h_1(-3e_2+5e_3))),\\
\frac{df_3}{dt_1}&=\frac{1}{6}(-3 f_3^2(h_1+h_2)-3 f_2 h_1e_1-2f_1 h_1e_2-f_1h_2e_2+5h_1e_1e_2+h_2e_1e_2\\&+3 f_1^2 e_3-3 f_2^2 e_3+h_1^2e_3-h_2^2 e_3+f_3(-h_1^2+h_2^2-3f_1e_1+3f_2 e_2+3h_1 e_3+3 h_2e_3)),\\
\frac{dh_1}{dt_1}&=\frac{1}{6}(3 f_1^2 h_1+f_3^2 h_1-f_3^2 h_2+3 f_3h_1h_2-2f_2^2(2 h_1+h_2)+3h_1e_1^2+4 h_1e_2^2\\
&+2 h_2e_2^2-3 h_1h_2e_3+3 f_2e_1e_3-h_1e_3^2+h_2e_3^2+3 f_1(f_2 f_3-2 e_2 e_3)),\\
\frac{dh_2}{dt_1}&= \frac{1}{6}(3 f_1^2 h_1+f_3^2 h_1-f_3^2 h_2+3 f_3 h_1h_2-2 f_2^2(2 h_1+h_2)+3 h_1e_1^2\\
&+4 h_1 e_2^2+2h_2e_2^2-3 h_1h_2e_3+3 f_2e_1e_3-h_1e_3^2+h_2e_3^2+3 f_1(f_2 f_3-2 e_2e_3)),\\
\frac{de_1}{dt_1}&=\frac{1}{18}(9 f_1^2(f_2-2 e_1)+f_1(-5 h_1^2-2 h_1 h_2+4 h_2^2+18 e_1^2+6 f_2 e_2-3 f_3(2 h_1+7 h_2-2 e_3)\\
&-18 e_3^2)+3(h_1^2 e_1+e_2(-3 f_2 e_1+2 h_2e_3)+2 h_1(h_2e_1+2e_2e_3)-3 f_3(h_1(e_1+e_2)-e_1(2h_2+e_3)))),\\
\frac{de_2}{dt_1}&=\frac{1}{6}(-3 f_1^2 f_2-2 f_3h_1e_1-f_3h_2 e_1+f_1(f_3(-h_1+h_2)+3 e_1(2f_2-e_2))+6 f_3h_1e_2\\
&-3 f_3h_2e_2-2 h_1h_2 e_2-h_2^2e_2+3h_1e_1e_3-3 f_3e_2e_3+f_2(2h_1h_2+h_2^2+f_3(-6 h_1+3 h_2)+3e_3^2)),\\
\frac{de_3}{dt_1}&=\frac{1}{6}(-3 f_1^2h_1-3 f_3^2(h_1+h_2)+2 f_2h_1 e_1+f_2h_2e_1+h_1e_3-h_2^2e_3-3f_2e_2 e_3\\
&+f_3(-h_1^2+h_2^2+3 f_2 e_2+3h_1e_3+3 h_2e_3)- f_1(f_2(2 h_1+h_2)+3f_3e_1-3 h_1e_2-3 e_1e_3)).
\end{align*}
\bibliographystyle{jalpha}
\bibliography{math}

\begin{thebibliography}{oGVAG04}
\bibitem[Ad]{Ad} Adler, M.; On a trace functional for formal pseudo differential operators and the symplectic structure of the Korteweg-de\thinspace Vries type equations. Invent. Math. 50 (1978/79), no. 3, 219--248.
\bibitem[Ara1]{Ara1} Arakawa T.; A remark on the $C_2$-cofiniteness condition on vertex algebras, Math. Z. (2012) 270:559--575.
\bibitem[Ara2]{Ara2} Arakawa T.; Associated varieties of modules over Kac-Moody algebras and $C_2$-cofiniteness of W-algebras. Int. Math. Res. Not., 2015:11605--11666, 2015.
\bibitem[Ara3]{Ara3} Arakawa T.; Rationality of W-algebras: principal nilpotent cases. Ann. Math., 182(2):565--694, 2015.
\bibitem[Ar1]{Ar1} Arnold V. I. ; Mathematical Methods of Classical Mechanics,(Graduate Texts in Mathematics, 60), May 16, 1989.
\bibitem[AraM]{AraM} Arakawa T., Moreau A; Arc spaces and chiral symplectic cores. Publ. Res. Inst. Math. Sci. 57 (2021), no. 3--4, 795--829.
\bibitem[AraM1]{AraM1} Arakawa T., Moreau A.; Arc spaces and vertex algebras, Monograph, April 17, 2021.
\bibitem[AN1]{AN1} Arnold V. I., S. P. Novikov; Dynamical Systems VII: Integrable Systems Nonholonomic Dynamical Systems, Encyclopaedia of Mathematical Sciences (EMS, volume 16).
\bibitem[AM]{AM} Adler M., van Moerbeke P.; Completely integrable systems, Euclidean Lie algebras, and curves, Volume 38, Issue 3, December 1980, Pages 267--317.
\bibitem[B1]{B1} Belavin A.; KdV-Type Equations and W-Algebras, Advanced Studies in Pure Mathematics 19, 1989, Integrable Systems in Quantum Field, Theory and Statistical Mechanics, pp. 117--125. 
\bibitem[Ba1]{Ba1} Bateman H.; higher transcendental functions, vol 2, McGraw Hill 1953.
\bibitem[BD]{BD} Belavin A, Drinfel'd V.; Triangle Equations and Simple Lie Algebras, Harwood Academic, 1998.
\bibitem[BDSK09]{BDSK09} Barakat A., De Sole A., Kac V. G.; Poisson vertex algebras in the theory of Hamiltonian equations, Jpn. J. Math. 4 (2009), no. 2, 141252.
\bibitem[Ber1]{Ber1} Bertram K.; The solution to a generalized Toda lattice and representation theory. Adv. in Math. 34 (1979).
\bibitem[BK03]{BK03} Bakalov B., Kac V. G.; Field algebras, IMRN, Volume 2003, Issue 3, 2003, Pages 123--159. 
\bibitem[Bo1]{Bo1} Borcherds R.; Vertex algebras, Kac-Moody algebras, and the Monster, Proc. Nat. Acad. Sci. U.S.A. 83 (1986) 3086--3071.
\bibitem[BPZ]{BPZ} Belavin A., Polyakov A., Zamolodchikov A.; Infinite conformal symmetry in two-dimensional quantum field theory, Nuclear Physics B, Volume 241, Issue 2, 23 July 1984.
\bibitem[BS]{BS} Beisert N., Spill F.; The Classical r-Matrix of AdS/CFT and its Lie Bialgebra Structure, Communications in Mathematical Physics volume 285, pages 537--565 (2009).
\bibitem[Ca1]{Ca1} Casati M.; On deformations of multidimensional Poisson brackets of hydrodynamic type. Comm. Math. Phys. 335 (2015), no. 2, 851--894. 
\bibitem[Ca15]{Ca15} Casati M.; Multidimensional Poisson Vertex Algebras and Poisson cohomology of Hamiltonian structures of hydrodynamic type, SISSA-Trieste, Doctoral Thesis, 2015.
\bibitem[Car]{Car} Carter R.; Lie Algebras of Finite and Affine Type, Cambridge Studies in Advanced Mathematics, Series Number 96, October 27, 2005. 
\bibitem[ChaP]{ChaP}  Chari V., Pressley A. N.; A Guide to Quantum Groups, Cambridge University Press, July 27, 1995.
\bibitem[CM]{CM} Chervov, A. V., Molev, A.I.: On higher-order Sugawara operators. Int. Math. Res. Not. (9),16121635 (2009).
\bibitem[CWY1]{CWY1} Costello K.; Witten Ed., Yamazaki M.; Gauge Theory and Integrability, I, ICCM Not. 6, 46--119 (2018).
\bibitem[CWY2]{CWY2} Costello K., Edward Witten, Masahito Yamazaki: Gauge Theory and Integrability, II, ICCM Not. 6, 120--146 (2018).
\bibitem[DG]{DG} Gelfand I. M., Dikii L. A.; Asymptotic properties of the resolvent of Sturm-Liouville equations, and the algebra of Korteweg-de Vries equations.Uspehi Mat. Nauk, 30(5(185)):67--100, 1975.
\bibitem[DKT]{DKT} De Sole A.; Kac V. G.; Turhan, Refik: A new approach to the Lenard-Magri scheme of integrability, Comm. Math. Phys. 330 (2014), no. 1, 107--122.
\bibitem[DKV]{DKV} De Sole A.; Kac V. G.; Valeri D.: A new scheme of integrability for (bi)Hamiltonian PDE. Comm. Math. Phys. 347 (2016), no. 2, 449--488.
\bibitem[DSK06]{DSK06} De Sole A., Kac V. G.; Finite vs. affine W-algebras, Japan. J. Math. 1(2006), no.1, 137--261.
\bibitem[DSKV13]{DSKV13} De Sole A., Kac V. G., Daniele V.:  Classical $\mathcal{W}$-Algebras and Generalized Drinfeld-Sokolov Bi-Hamiltonian Systems Within the Theory of Poisson Vertex Algebras, Communications in Mathematical Physics volume 323, pages 663--711 (2013).
\bibitem[DSKD22]{DSKD22} De Sole A., Kac V. G., Daniele V.; On {L}ax operators, Jpn. J. Math., Vol 17, 2022.
\bibitem[DN1]{DN1} Dubrovin B. A., Novikov S. P.; Hamiltonian formalism of one-dimensional systems of the hydrodynamic type and the Bogolyubov Whitham averaging method. Dokl. Akad. Nauk SSSR, 270(4):781--785,
1983.
\bibitem[DN2]{DN2} Dubrovin B. A., Novikov S. P.; Poisson brackets of hydrodynamic type. Dokl. Akad. Nauk SSSR, 279(2):294--297, 1984.
\bibitem[DN3]{DN3} Dubrovin B. A., Novikov S. P., Fomenko A. T.; Modern Geometry-Methods and Applications Part II: The Geometry and Topology of Manifolds, Graduate Texts in Mathematics (GTM, volume 104), 1985.
\bibitem[Dun]{Dun} Dunajski M.; Solitons, Instantons, And Twistors, Oxford Graduate Texts In Mathematics (2010).
\bibitem[DZ]{DZ} Dubrovin B., Zhang Y.: Normal forms of hierarchies of integrable PDEs, Frobenius manifolds and Gromov-Witten invariants, 	SISSA 65/2001/FM, 2001.
\bibitem[FF]{FF} Feigin B, Frenkel Ed; Affine Kac-Moody algebras at the critical level and Gelfand-Dikii algebras. In Infinite analysis, Part A, B (Kyoto, 1991), volume 16 of Adv. Ser. Math. Phys., pages 197--215. World Sci. Publ., River Edge, NJ, 1992.
\bibitem[FBZ]{FBZ} Frenkel Ed., Ben-Zvi D.; Vertex algebras and algebraic curves 2nd Edition Mathematical Surveys and Monographs Volume: 88, 2004.
\bibitem[Fre]{Fre} Frenkel Ed.; Langlands correspondence for loop groups, volume 103 of Cambridge Studies in Advanced Mathematics. Cambridge University Press, Cambridge, 2007.
\bibitem[FZ]{FZ} Faddeev L. D., Zakharov V. E.; Korteweg-de Vries equation: A completely integrable Hamiltonian system, Functional Analysis and Its Applications, 1971.
\bibitem[FZh]{FZh} Frenkel I. B., Zhu Y.; Vertex operator algebras associated to representations of affine and {V}irasoro algebras, Duke Math. J., Vol.66, 1992, No.1, p.g. 123--168.
\bibitem[GFP]{GFP} Gregorio F., Franco M.; Pedroni M.: Bihamiltonian geometry and separation of variables for Toda lattices. Nonlinear evolution equations and dynamical systems  J. Nonlinear Math. Phys. 8 (2001).
\bibitem[GT]{GT} Gradshteyn R., Tseytlin G. et al; Table of Integrals, Series, and Products, Academic Press, 2015.
\bibitem[Ha1]{Ha1} Hayashi T.; Sugawara operators and Kac-Kazhdan conjecture, Inventiones mathematicae, volume 94, pages 13--52 (1988).
\bibitem[K1]{K1} Kac V. G.; Introduction to vertex algebras, Poisson vertex algebras, and integrable Hamiltonian PDE, Lecture one in the Perspectives in Lie Theory (Springer INdAM Series, 19), December 29, 2017.
\bibitem[KBI]{KBI} Korepin V. E., Bogoliubov N. M., Izergin A. G.; Quantum inverse scattering method and correlation functions, Cambridge university press, 1997.
\bibitem[Kir]{Kir1} Kirillov A. A.; Lectures on the Orbit Method, Graduate Studies in Mathematics Volume 64, June 1, 2004. 
\bibitem[Li1]{Li1} Li H.-S.; Vertex algebras and vertex Poisson algebras Commun. Contemp. Math. 6 (2004).
\bibitem[Li2]{Li2} Li H.-S.: Abelianizing vertex algebras. Comm. Math. Phys., 259(2): 391--411, 2005.
\bibitem[LL]{LL} Lepowsky J., Li H.-S.; Introduction to Vertex Operator Algebras and Their Representations (Progress in Mathematics Book 227), 2004.
\bibitem[LL1]{LL1} Landau L., Lifshitz E.; Course of Theoretical Physics Vol I, Butterworth-Heinemann, January 1, 1982.
\bibitem[LM]{LM} Linshaw A., Malikov F.; One example of a chiral Lie group, 	arXiv:1902.07414, 2019. 
\bibitem[Lu]{Lu} Lu J. H.; Multiplicative and Affine Poisson
Structures on Lie Groups, Ph.D. dissertation, 1990.
\bibitem[M1]{M1} Michaelis W.; Lie coalgebras, Ph.D. Dissertation, University of Washington, Seattle, 1974.
\bibitem[Mo1]{Mo1} Molev A.; Sugawara Operators for Classical Lie Algebras (Mathematical Surveys and Monographs) March 15, 2018.
\bibitem[Nou]{Nou} Noumi M.; Painleve Equations Through Symmetry, Translations of Mathematical Monographs, 2004.
\bibitem[NZ]{NZ} Novikov S. P., Zakharov V.E.; Theory of Solitons: The Inverse Scattering Method, Springer Science, May 1984.
\bibitem[RS]{RS} Reyman A. G., Semenov-Tyan-Shanskii M. A.; Reduction of Hamiltonian systems, affine Lie algebras, and Lax equations, I, II, Invent. Math.,54, No. I, 81--100 (1979), and 63, No. 3, 423--432 (1981)
\bibitem[Sch]{Sch} Schottenloher M.; A mathematical introduction to conformal field theory, Lecture Notes in Physics (LNP, volume 759), 2008.
\bibitem[Se1]{Se1} Semenov-Tyan-Shanskii M. A.; Integrable Systems: the $R$-matrix Approach RIMS-1650, 2008.
\bibitem[Se2]{Se2} Semenov-Tyan-Shanskii M. A.; What is a classical $R$-matrix?, Funct Anal Its Appl 17, 259--272 (1983).
\bibitem[TU]{TU} Terng C-L.; Uhlenbeck K.; Poisson actions and scattering theory for integrable systems. Surveys in differential geometry: integral systems [integrable systems], 315--402, Surv. Differ. Geom., 4, Int. Press, Boston, MA, 1998.
\bibitem[WW]{WW} Whittaker E. T., Watson G. N.; A Course of Modern Analysis Cambridge Mathematical Library (1996).
\bibitem[Xu1]{Xu1} Xu X.: Classical R-matrices for vertex operator algebras Journal of Pure and Applied Algebra 85 (1993) 203--218.
\bibitem[YKSch]{YKSch} Kosmann-Schwarzbach Y.; Lie bialgebras, Poisson Lie groups and dressing transformations, Integrability of Nonlinear Systems, Second edition, Lecture Notes in Physics 638, Springer-Verlag, 2004, pp. 107--173.
\end{thebibliography}

\end{document}